\documentclass[onecolumn]{IEEEtran}
\usepackage{basicreq}

\usepackage[top=1.5in, bottom=1.25in, left=1.25in, right=1.25in]{geometry}

\def\L{\protect\Lambda}
\def\Ln{\Lambda^{(n)}}
\def\Dta{\Delta}
\def\cB{\mathcal{B}}
\def\bE{\mathbb{E}}

\def\hH{\widehat{H}}
\def\lba{\lambda}

\newcommand{\nchoosek}[2]{\begin{pmatrix}#1 \\#2 \end{pmatrix}}
\def\Supp{\mathrm{Supp}}
\def\Fp{\mathbb{F}_p}
\def\cG{\mathcal{G}}
\def\cA{\mathcal{A}}

\def\dta{\delta}
\def\Su{\mathbb{S}(\u)}
\def\tcol{}

\title{Some ``Goodness'' Properties of LDA Lattices}

\author{Shashank Vatedka and Navin Kashyap
\thanks{This work was presented in part at the 2015 IEEE Information Theory Workshop at Jerusalem, Israel.}
\thanks{Shashank Vatedka and Navin Kashyap are with the Dept.\ of Electrical Communication Engineering,
Indian Institute of Science, Bengaluru, India.
Email: {$\{$shashank,nkashyap$\}$@ece.iisc.ernet.in}}}

\begin{document}


 \maketitle
 
 \begin{abstract}
 We study some structural properties of Construction-A lattices obtained from low density parity check (LDPC) codes over prime fields.
 Such lattices are called low density Construction-A (LDA) lattices, and permit low-complexity belief propagation decoding for 
 transmission over Gaussian channels. It has been shown that LDA lattices achieve the capacity of the power constrained additive white Gaussian noise (AWGN) channel  with closest lattice-point decoding,
 and simulations suggested that they perform well under belief propagation  decoding. We continue this line of work,
 and prove that these lattices are good for packing and mean squared error (MSE) quantization, and that their duals are good for packing.
 With this, we can conclude that codes constructed using nested LDA lattices can achieve the capacity of the power constrained AWGN channel, the capacity of the dirty paper channel,
 the rates guaranteed by the compute-and-forward protocol, and the best known rates for bidirectional relaying with perfect secrecy.
 \end{abstract}

 \section{Introduction}\label{sec:intro}

 Nested lattice coding for communication over Gaussian networks has received considerable attention in recent times. 
 It has been shown~\cite{Erez04} that nested lattice codes with closest lattice-point decoding can achieve the capacity of the power constrained additive white Gaussian noise 
 (AWGN)
 channel. They are also known to achieve the capacity of the dirty-paper channel~\cite{Erez_dpaper05}. 
 Inspired by these results, they have been applied to design protocols for reliable communication over wireless Gaussian networks. 
 They have been used with much success for the interference channel~\cite{Bresler10,Vishwanath10}, the Gaussian bidirectional relay channel~\cite{Wilson,Nazer11}, and generalized to the problem of physical layer network coding~\cite{Baik,Nazer11}
 for multiuser Gaussian channels. 
 Nested lattice coding has also been used for security in wiretap channels~\cite{Belfiore10,Ling13} and bidirectional
 relay networks~\cite{HeYenerstrong,Vatedka13}.
 For a more comprehensive treatment of lattices and their applications in communication problems, see~\cite{Zamirbook}.
 
 Constructing lattices that have good structural properties is a problem that has been studied for a long time. 
 Poltyrev~\cite{Poltyrev} studied lattices
 in the context of coding for reliable transmission over the AWGN channel without power constraints, and showed that
 there exist lattices which are ``good'' for AWGN channel coding, i.e., achieve a vanishingly small probability of error for all
 sufficiently small values of the noise variance. 
 In addition to coding for the AWGN channel, lattices were also studied in prior literature in the context of several other problems such as sphere packing, sphere covering, and MSE quantization.
 In the sphere packing problem, we want to find an arrangement of non-intersecting spheres of a given radius that maximizes the average number of spheres packed per unit volume.
 On the other hand, the covering problem asks for an optimal covering of space by spheres of a given radius, that minimizes the average number of spheres per unit volume.
 In the MSE quantization problem, we want to find a minimal set of codewords which will ensure that the 
 average mean squared error/distortion is less than a specified quantity. The use of lattices to generate good sphere packings, sphere coverings, and quantizers is a well-studied problem~\cite{Conway,Zamirbook}.
 
 Finding lattices with good stuctural properties is of particular importance in designing lattice codes 
 that use nested lattice shaping for power constrained Gaussian channels. 
 A poorly designed shaping region leads to loss in transmission rates.
 It was shown in~\cite{Erez04} that using nested lattice codes,
 where the fine lattices are good for AWGN channel coding (in the sense of Poltyrev's definition) and the coarse lattices are good for MSE quantization, we can achieve the capacity
 of the power constrained AWGN channel. Furthermore, the rates guaranteed by~\cite{Wilson, Nazer11} for bidirectional relaying and the compute-and-forward protocol are achievable using nested lattices that satisfy
 the aforementioned properties.  It was shown that if in addition to the above properties, the duals of the coarse lattices
 are also good for packing, then a rate of $\frac{1}{2}\log_2 \text{SNR}-\log_2(2e)$ (where SNR denotes the signal-to-noise ratio) can be achieved with perfect (Shannon) secrecy over the bidirectional relay~\cite{Vatedka13}.
 
 Instead of studying arbitrary lattices, it is easier to study lattices that have a special structure, i.e., 
 lattices constructed by lifting a linear code over a prime field to $\R^n$. One such technique to obtain
 lattices from linear codes is Construction A~\cite{Conway}, where the lattice is obtained by tessellating
 the codewords of the linear code (now viewed as points in $\R^n$) across the Euclidean space. It was shown in~\cite{Erez05}
 that if we pick a linear code uniformly at random, then the resulting Construction-A lattice is
 asymptotically good for covering, packing, MSE quantization, and AWGN channel coding with high probability.
 
 The problem with general Construction-A lattices is the complexity of closest lattice-point decoding. There is no known
 polynomial-time algorithm for decoding Construction-A lattices obtained from arbitrary linear codes. A natural way of circumventing this is to restrict 
 ourselves to LDPC codes to construct lattices.
 We can then use low-complexity belief propagation (BP) decoders instead of the closest lattice-point decoder which has exponential complexity.
 Such lattices, termed low-density Construction-A (LDA) lattices, were introduced in in~\cite{lda_itw12}.
 Simulation results in~\cite{diPietrothesis,Tunali13} showed that these lattices perform well
 with BP decoding. While there is no formal proof that these lattices are good under BP decoding, it was proved in~\cite{lda_ita13} that LDA lattices are good for AWGN channel coding, and subsequently shown in~\cite{diPietrothesis,diPietro_journal} that nested LDA lattices achieve
 the capacity of the power constrained AWGN channel with
 closest lattice-point decoding. In this paper, we show that LDA lattices have several other goodness properties. 
 We will prove that a randomly chosen LDA lattice (whose parameters satisfy certain conditions) is good for packing and MSE quantization with probability tending to $1$ as $n\to\infty$.
 In addition, we will show that the dual of a randomly chosen LDA lattice is good for packing with probability tending to $1$ as $n\to\infty$. 
This means that the capacities of the power constrained AWGN channel and the dirty paper channel, the rates guaranteed by compute-and-forward framework~\cite{Nazer11},
and the rates guaranteed by~\cite{Vatedka13} for perfectly secure bidirectional relaying can all be achieved using nested LDA lattices 
(with closest lattice-point decoding). However, showing that the aforementioned results can all be achieved using belief propagation
decoding still remains an open problem.
Even though other AWGN-good lattice constructions that permit low-complexity decoding algorithms have been proposed~\cite{Sommer, Yan13}, this is the first instance where such a class of lattices
have been shown to satisfy other goodness properties, and this is the main contribution of this work.

The rest of the paper is organized as follows: We describe the notation and state some basic definitions in the next two subsections. Section~\ref{sec:LDA_ensemble} describes the 
ensemble of lattices, and the main result is stated in Theorem~\ref{theorem:LDA_simultaneousgoodness}. Some preliminary lemmas are stated in Section~\ref{sec:prelim_lemmas}. 
This is then followed by results on the various goodness properties of 
lattices in the LDA ensemble. In Section~\ref{sec:channelcoding}, the goodness of these lattices for channel coding is described. This is followed by
Section~\ref{sec:LDApacking} on the packing goodness of LDA lattices. In Section~\ref{sec:LDA_MSEquantization}, we discuss sufficient conditions
for goodness of these lattices for MSE quantization. We then prove the goodness of the duals for  packing in Section~\ref{sec:LDA_dualpacking},
and conclude with some final remarks in Section~\ref{sec:remarks}. Some of the technical proofs are given in the appendices.

\section{Notation and Basic Definitions}
\subsection{Notation}\label{sec:notation}
The set of integers is denoted by $\Z$, and the set of reals by $\R$. For a prime number $p$, the symbol $\Fp$ denotes the field of integers modulo $p$.
Matrices are denoted by uppercase letters, such as $A$, and column vectors by boldface lowercase letters, such as $\u$. The $\ell^2$ (or Euclidean) norm of
a vector $\u$ is denoted by $\Vert \u \Vert$. 
The support of a vector $\u$ is the set of all coordinates of $\u$ which are not zero, and is denoted by $\mathrm{Supp}(\u)$.
If $\mathcal{A}$ is a finite set, then $|\mathcal{A}|$ is the number of elements in $\mathcal{A}$.
The same notation is used for the absolute value of a real number $r$ ($|r|$), but the meaning should be clear from the context. If $\cA$ and $\cB$ are two subsets of $\R^n$,
and $\alpha, \beta$ are real numbers, then $\alpha\cA+\beta\cB$ is defined to be $\{ \alpha \x+\beta\y:\x\in \cA, \y\in\cB \}$. Similarly, for  $\x\in \R^n$, 
we define $\x+\alpha \cB=\{\x+\alpha\y: \y\in\cB   \}$.

We define $\cB$ to be the (closed) unit ball in $n$ dimensions centered at $\0$.
For $\x\in \R^n$, and $r>0$,  the $n$ dimensional closed ball in $\R^n$ centered at $\x$ and having radius $r$ is denoted by $r\cB+\x:=\{ r\u+\x:\u\in \cB \}$.
We also define $V_n:=\text{vol}(\cB)$,  the volume of a unit ball in $n$ dimensions. 

For $0\leq a\leq 1$, $h_2(a):=-a\log_2 a-(1-a)\log_2(1-a)$ denotes
the binary entropy of $a$. 
If $f(n)$ is a sequence indexed by $n\in \{ 1,2,3,\ldots \}$, then we say that $f(n)=o(1)$ if $f(n)\to 0$ as $n\to\infty$.

\subsection{Basic Definitions}\label{sec:basicdefs}

\begin{figure}
\begin{center} 
 \includegraphics[width=10cm]{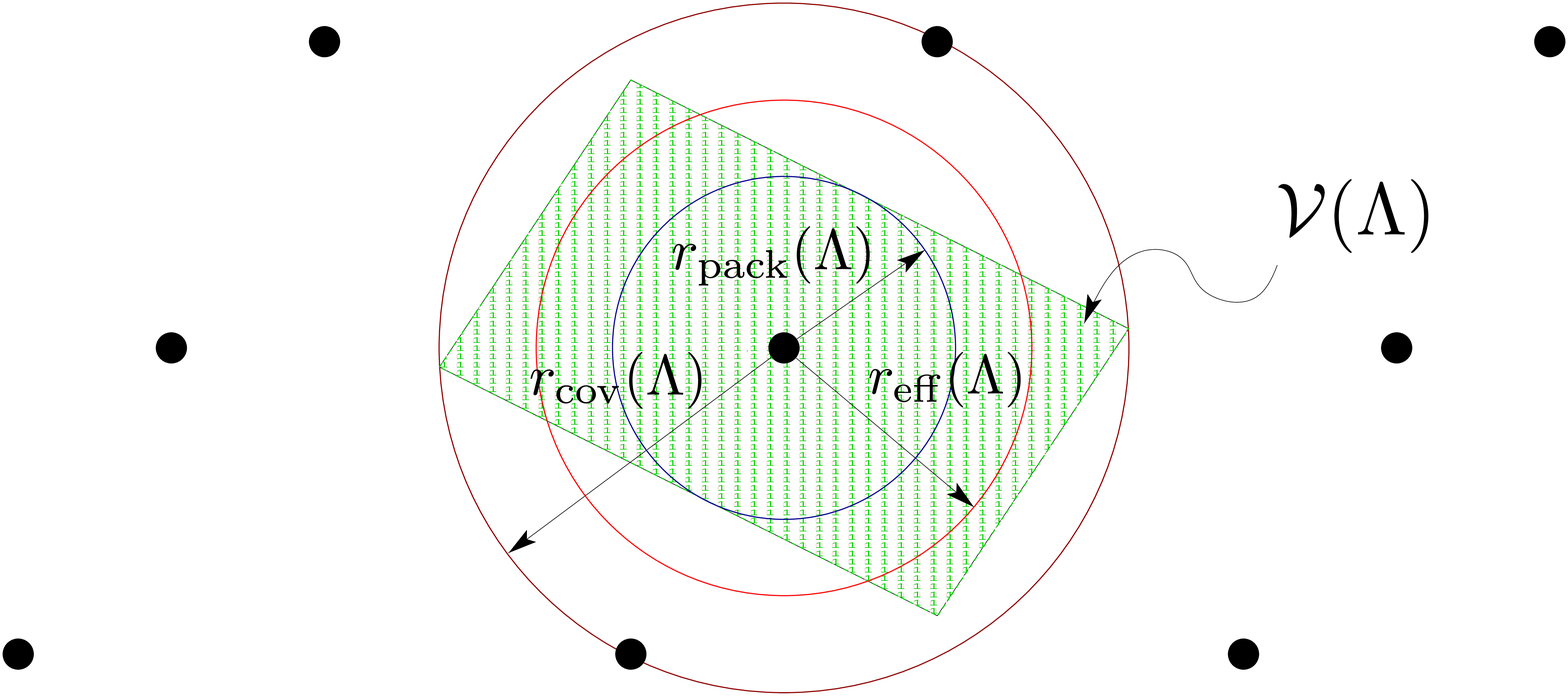}
\end{center}
\caption{Some important parameters of a lattice.}
\label{fig:latticeparam}
\end{figure}

We will state some basic definitions related to lattices. The interested reader is directed to~\cite{Erez05,Zamirbook} for more details.
Let $A$ be a full-rank $n\times n$ matrix with real-valued entries. Then, the set of all integer-linear combinations of the columns of $A$ forms an additive group and
is called an $n$-dimensional lattice, i.e., $\L=A\Z^n:=\{ A \x:\x\in \Z^n \}$. 
The matrix $A$ is called a \emph{generator matrix} for $\L$. The \emph{dual lattice} of $\L$, denoted by $\L^{*}$, is defined as $\L^*:=\{\y\in \R^n:\x^{T}\y\in \Z,\; \forall \x\in \L\}$. 
If $A$ is a generator matrix for $\L$, then $(A^{-1})^T$ is a generator matrix for $\L^*$. 

The set of all points in $\R^n$ for which the zero vector is the closest lattice point (in terms of the $\ell^2$ norm), with ties decided according to a 
fixed rule, is called the \emph{fundamental Voronoi region},
and is denoted by $\cV(\L)$. The set of all translates of $\cV(\L)$ by points in $\L$ partitions $\R^n$ into sets called \emph{Voronoi regions}. 

The \emph{packing radius} of $\L$, $\rpack(\L)$, is the radius of the largest $n$-dimensional open ball
that is contained in the fundamental Voronoi region. The \emph{covering radius} of $\L$, $\rcov(\L)$, is the radius of the smallest closed ball that contains $\cV(\L)$. 
Let $\vol(\L)$ be the volume of the fundamental Voronoi region. 
Then, the \emph{effective radius} of $\L$ is defined to be the radius of the $n$-dimensional ball having volume $\vol(\L)$, and is denoted by $\reff(\L)$.
These parameters are illustrated for a lattice in two dimensions in Fig.~\ref{fig:latticeparam}.

If $\L,\L_0$ are $n$-dimensional lattices satisfying $\L_0\subset \L$, then $\L_0$ is said to be \emph{nested} within $\L$, or $\L_0$ is called a \emph{sublattice} of $\L$.
The lattice $\L$ is called the \emph{fine lattice}, and $\L_0$ is called the \emph{coarse lattice}.
The quotient group $\L/\L_0$ has
\[
    |\L/\L_0|=\frac{\text{vol}(\L_0)}{\text{vol}(\L)}
\]
elements, and the above quantity is called the \emph{nesting ratio}. This is equal to the number of points of $\L$ within $\cV(\L_0)$.

We now formally define the ``goodness'' properties that we want lattices to satisfy. A sequence of lattices, $\{\Ln \}$ (indexed by the dimension, $n$),
is \emph{good for packing} if\footnote{
   \tcol{
The definition of packing goodness is derived from the best known lower bound of $1/2$ for the asymptotic value of $\rpack(\L)/\reff(\L)$ due to Minkowski~\cite{Minkowski} and Hlawka~\cite{Hlawka}. See~\cite{Erez05} for a discussion of the same, and~\cite{Rogersbook} for more details regarding the sphere packing problem.
   }
} 
\[
 \limsup_{n\to\infty}\frac{\rpack(\Ln)}{\reff(\Ln)}\geq \frac{1}{2}.
\]

Lattices have been well-studied in the context of vector quantization, where the aim is to
obtain a codebook of minimum rate while ensuring that the average distortion (which is the mean squared error in this case)
is below a threshold.
The \emph{normalized second moment per dimension} of an $n$-dimensional lattice $\L$ is defined as
\begin{equation}
 {G}(\L)=\frac{1}{n\left(\text{vol}(\L) \right)^{1+2/n}}\int_{\cV(\L)}\Vert \y \Vert^2\: d\y.
\label{eq:Glambda}
\end{equation}
This is equal to the normalized second moment of a random variable (the error vector in the context of quantization) which is uniformly distributed
over the fundamental Voronoi region of $\L$, and we want this to be as small as possible.
The normalized second moment of any lattice is bounded from below by that of an $n$-dimensional sphere,
which is equal to $1/(2\pi e)$ (see e.g.,~\cite{Erez05}).
A sequence of lattices $\{\Ln\}$ is said to be \emph{good for MSE quantization} if ${G}(\Ln)\to\frac{1}{2\pi e}$ as $n\to\infty$.

We also want to use lattices to design good codebooks for reliable transmission over additive noise channels.
Classically, a lattice was defined to be good for AWGN channel coding~\cite{Erez05} if with high probability, the closest lattice-point decoder returned the 
actual lattice point that was transmitted over an AWGN channel without power constraints. This notion was made slightly more general in~\cite{Ordentlich}, 
using the notion of semi norm-ergodic noise:
\begin{definition}[\cite{Ordentlich}]
 A sequence of random vectors $\{\z^{(n)}\}$ (where $\z^{(n)}$ is an $n$-dimensional random vector) having  second moment per dimension $\sigma^2:=\frac{1}{n}\mathbb{E}[\Vert\z^{(n)}\Vert^2]$ for all $n$, is said to be semi norm-ergodic if
 for every $\delta>0$,
 \[
  \Pr[\z^{(n)}\notin (\sqrt{(1+\delta)n\sigma^2})\cB]\to 0\text{ as }n\to\infty.
 \]
 \label{def:seminorm_ergodic}
\end{definition}

As remarked in~\cite{Ordentlich}, any  zero-mean noise  whose components are independent and identically distributed (i.i.d.\ ) is semi norm-ergodic.
We say that a sequence of  lattices $\{ \Ln \}$ is \emph{good for coding in presence of semi norm-ergodic noise} if for every sequence of semi norm-ergodic noise vectors $\{\z^{(n)}\}$,  with second moment per dimension
equal to $\sigma^2:= \frac{1}{n}\mathbb{E}[\Vert\z^{(n)}\Vert^2]$, the probability that the lattice point closest to $\z^{(n)}$ is not $\0$ goes to zero as
$n\to\infty$, i.e.,
\[
 \Pr[\z^{(n)}\notin \cV(\Ln)]\to 0 \text{ as }n\to\infty,
\]
as long as $(\text{vol}(\Ln))^{2/n}>2\pi e \sigma^2$ for all sufficiently large $n$.
\tcol{  Similarly, we say that a sequence of lattices $\{ \Ln \}$ is \emph{good for AWGN channel coding} if for every sequence of  noise vectors $\{\z^{(n)}\}$, with independent and identically distributed (i.i.d.) Gaussian components with mean $0$ and variance $\sigma^2$, the probability that the lattice point closest to $\z^{(n)}$ is not $\0$ goes to zero as
$n\to\infty$, i.e.,\footnote{Note that this is weaker than the definition used in e.g.,~\cite{Erez04,Erez05}, where $\Pr[\z^{(n)}\notin \cV(\Ln)]$ is required to go to zero exponentially in $n$.}   
\[
 \Pr[\z^{(n)}\notin \cV(\Ln)]\to 0 \text{ as }n\to\infty,
\]
as long as $(\text{vol}(\Ln))^{2/n}>2\pi e \sigma^2$ for all sufficiently large $n$.  }
 
 An LDPC code can be defined by its parity check matrix, or by the corresponding edge-labeled Tanner graph~\cite{RichardsonBook}. 
 A $(\Dta_V,\Dta_C)$-regular bipartite graph $\cG=((V,C),\mathcal{E})$ is defined as an undirected bipartite graph with every left vertex (i.e., every vertex in $V$)
 having degree $\Dta_V$, and every right vertex (i.e., every vertex in $C$) having degree $\Dta_C$. 
 The vertices in $V$ are also called the variable nodes, and those in $C$ are called parity check (or simply, check) nodes.
\tcol{  The graph $\mathcal{G}$ is the Tanner graph of a binary linear code with parity check matrix $\hH$.
The matrix $\hH$ has entries from $\{ 0,1 \}$, and the $(i,j)$th entry is $1$ if and only if there is an edge in $\cG$ between $i$ and $j$. 
}
 If $\mathcal{A}$ is a subset of $V$ (resp.\ $\mathcal{A}'\subset C$), then $N(\mathcal{A})$ is the neighbourhood of $\mathcal{A}$, defined as
 $N(\mathcal{A}):=\{ v\in C: (u,v)\in \mathcal{E} \text{ for some }u\in \mathcal{A}\}$ (resp.\ $N(\mathcal{A'}):=\{ u\in V: (u,v)\in \mathcal{E}\text{ for some }v\in \mathcal{A}' \}$).

 \section{The Ensemble of LDA Lattices}\label{sec:LDA_ensemble}
 
 Throughout this paper, $\lba$ and $R$ are real numbers chosen so that $\lba>0$, and $1> R > 0$.
For $n\in \Z^+$, define $k:=\lceil nR\rceil$. 
For each $n\in\Z^+$, let $p$ (which is a sequence indexed by $n$) be the smallest prime number greater than or equal to $n^\lba$, and $\Fp$ denote the field of integers modulo $p$.

We study the constant-degree LDA ensemble introduced in~\cite{lda_ita13,diPietrothesis}. 
Specifically, let $\mathcal{G}$ denote a $(\Dta_V,\Dta_C)$-regular bipartite graph ($\Dta_V<\Dta_C$), with $n$ variable nodes,  $\frac{n\Dta_V }{\Dta_C}$ check nodes, and satisfying $R=1-(\Dta_V/\Dta_C)$. 
The graph $\cG$ is required to satisfy certain expansion properties, which are stated in the definition below.

\begin{definition}[\cite{diPietrothesis}, Definition 3.3]
 Let $A,\alpha,B,\beta$ be positive real numbers satisfying $1\leq\alpha<A$, and $\frac{1}{1-R}<\beta<\min\{ \frac{2}{1-R}, B \}$. 
 \tcol{  Let $\epsilon$ and $\vartheta$ be two small positive 
constants satisfying $\epsilon<(1-R)/A$ and $\vartheta<1/(B(1-R))$.
}
The graph $\cG$ is said to be $(\alpha, A,\beta,B,\epsilon,\vartheta)$-good if 
\begin{enumerate}
 \item[(L1)] If $S\subset V$, and $|S|\leq \lceil \epsilon n\rceil$, then $|N(S)|\geq A|S|$.
 \item[(L2)] If $S\subset V$, and $|S|\leq \left\lceil \frac{n(1-R)}{2\alpha} \right\rceil$, then $|N(S)|\geq \alpha |S|$.
 \item[(R1)] If $T\subset C$, and $|T|\leq \vartheta n(1-R)$, then, $|N(T)|\geq B|T|$.
 \item[(R2)] If $T\subset C$, and $|T|\leq \frac{n(1-R)}{2}$, then $|N(T)|\geq \beta |T|$.
\end{enumerate}
\label{defn:goodgraphs}
\end{definition}
\tcol{
We call an infinite sequence of $(\Delta_V,\Delta_C)$ graphs $\{\cG^{(n)}:n\in \mathcal{I}\subset \Z^+\}$ to be ``superexpanders'' with parameters $(\alpha,A,\beta,B,\epsilon,\vartheta)$ if for each $n$ in the index set $\mathcal{I}$, we have $\cG^{(n)}$ being  $(\alpha, A,\beta,B,\epsilon,\vartheta)$-good.
}

\tcol{
The term ``expander graph'' was first coined by Bassalygo and Pinsker~\cite{BassalygoPinsker}, and Bassalygo~\cite{Bassalygo81} subsequently showed that random graphs are expanders. 
Expander graphs were used to construct codes for the binary symmetric channel in~\cite{Sipser_expandercode}, and they have been used to solve various problems in coding theory and computer science.
A survey of expander graphs and their applications can be found in~\cite{Expandersurvey}.
 We remark that we require more expansion properties that what is typically used in the literature. In most works, only the left expansion properties, i.e., (L1) and (L2) are used~\cite{Bassalygo81}, or $\Delta_V$ is chosen to be equal to $\Delta_C$, and $\vartheta=\epsilon, A=B, \alpha=\beta$~\cite{Expandersurvey}. However, we want both the variable nodes and the check nodes to expand, hence the term ``superexpander''.
      }

The following lemma by di~Pietro~\cite{diPietrothesis} asserts that a randomly chosen graph satisfies the above properties with high probability.

\begin{lemma}[\cite{diPietrothesis}, Lemma 3.3]
 Let $\cG$ be chosen uniformly at random from the standard ensemble~\cite[Definition~3.15]{RichardsonBook} of $(\Dta_V,\Dta_C)$-regular bipartite graphs with $n$ variable nodes.
Let $\epsilon$ and $\vartheta$ be positive constants.
If $\Dta_V$ satisfies
  \tcol{
\begin{align}
 \Dta_V > \max\Bigg\{ & \frac{h_2\left(\frac{1-R}{2\alpha}\right)+1-R}{h_2\left(\frac{1-R}{2\alpha}\right)-\frac{1}{2}h_2\left(\frac{1-R}{\alpha}\right)},R+2\alpha, A+1,\frac{h_2(\epsilon)+(1-R)h_2\left(\frac{A\epsilon}{1-R}\right)}{h_2(\epsilon)-\frac{A\epsilon}{1-R}h_2\left( \frac{1-R}{A} \right)}, &\notag \\
                           & \frac{1-R+h_2\left(\frac{\beta(1-R)}{2}\right)}{1-\frac{\beta(1-R)}{2}h_2\left(\frac{1}{\beta(1-R)}\right)}, \frac{(2+\beta R)(1-R)}{2-\beta(1-R)}, (1-R)(B+1), & \notag \\
                           & \frac{(1-R)h_2(\vartheta)+h_2(B\vartheta(1-R))}{h_2(\vartheta)-B\vartheta(1-R)h_2\left( \frac{1}{B(1-R)}\right)}, \frac{(A+1)(1-R)-A\epsilon (2-R)}{1-R-A\epsilon},&\notag\\
                           & \frac{B+1-\vartheta B(2-R)}{\frac{1}{1-R}-\vartheta B}\Bigg\} , &\label{eq:Dta_V_condns}
\end{align}
   }
then the probability that $\cG$ is not $(\alpha, A,\beta , B,\epsilon,\vartheta)$-good tends to zero as $n\to\infty$.
\label{lemma:goodgraphs}
\end{lemma}

 \subsection{{The $(\cG,\lba)$ LDA Ensemble}}
 Let $\lba>0$, and $1> R> 0$ be two constants, and $n\in\{ 1,2,3,\ldots \}$.
 Let $p$ be the smallest prime number greater than $n^\lba$.\footnote{In our proofs, we take $p=n^\lba$, and $k=nR$ for convenience, but choosing $p$ to be the smallest prime number greater than $n^\lba$,
 and $k=\lceil nR\rceil$
 will not change any of the results.} Let $\Dta_C:=\Dta_V/(1-R)$.
 Let us pick a $(\Dta_V,\Dta_C)$-regular bipartite graph $\cG$ with $n$ variable nodes.
 Throughout the paper, we assume that the parameters of $\cG$ satisfy the hypotheses of Lemma~\ref{lemma:goodgraphs}, and that $\cG$ is $(\alpha, A,\beta, B,\epsilon,\vartheta)$-good.
 Let $\hH$ denote the $n(1-R)\times n$ parity check matrix corresponding to the Tanner graph $\cG$.
 We describe the LDA ensemble obtained using the Tanner graph $\cG$, which will henceforth be called the $(\cG,\lba)$ LDA ensemble.
 
 We construct a new $n(1-R)\times n$ matrix, $H$, by replacing the $1$'s in $\hH$ with independent random variables uniformly distributed over $\Fp$.
For $1\leq i\leq n(1-R)$ and $1\leq j\leq n$,  let $h'_{i,j}$ be $n^2(1-R)$ i.i.d.\  random variables, each uniformly distributed over $\Fp$, and let $\widehat{h}_{i,j}$ be the $(i,j)$th entry of $\hH$. 
Then, the $(i,j)$th entry of $H$, denoted $h_{i,j}$, is given by $h_{i,j}=\widehat{h}_{i,j}h'_{i,j}$. Therefore, $h_{i,j}$ is equal to $h'_{i,j}$ if $\widehat{h}_{i,j}$ is $1$, and zero otherwise. 
For example, if
\[
  \hH=\begin{pmatrix}
     1 & 1 & 0 & 0 & 1 & 0\\
     0 & 1 & 0 & 1 & 0 & 1\\
     0 & 0 & 1 & 1 & 1 & 0\\
     1 & 0 & 1 & 0 & 0 & 1
    \end{pmatrix},
\]
then
\begin{equation}
  H=\begin{pmatrix}
     h_{11}' & h_{12}' & 0 & 0 & h_{15}' & 0\\
     0 & h_{22}' & 0 & h_{24}' & 0 & h_{26}'\\
     0 & 0 & h_{33}' & h_{34}' & h_{35}' & 0\\
     h_{41}' & 0 & h_{43}' & 0 & 0 & h_{46}'
    \end{pmatrix}.
\label{eq:H_eg}
\end{equation}

Note that the ``skeleton matrix'' $\hH$ is fixed beforehand, and the only randomness in $H$ is in the coefficients.
  \tcol{
Also observe that since $h_{ij}'$ is chosen uniformly at random from $\Fp$, the $(i,j)$th entry of $H$ is zero with positive probability even if  $\widehat{h}_{ij}=1$. 
However, if $p$ grows faster than $n$, then we can use the union bound to show that for a fixed $\widehat{H}$, the probability of having a zero coefficient is small, i.e., $\Pr[\exists (i,j) \text{ such that }h_{i,j}'=0 \text{ and }\widehat{h}_{ij}=1 ]\to 0$ as $n\to\infty$.
  }
The matrix $H$ is therefore the parity check matrix of an $n$-length $(\Dta_V,\Dta_C)$ regular LDPC code $\cC$ over $\Fp$ with high probability (if $\lba>1$). The LDA lattice $\L$ is obtained by applying Construction A to the code $\cC$,
i.e., $\L=\{\x\in \Z^n: \x\equiv \mathbf{c}\bmod p, \text{ for some } \mathbf{c}\in \cC \}$. Equivalently, if $\Phi$ denotes the natural embedding of $\Fp^n$ into $\Z^n$, then $\L=\Phi(\cC)+p\Z^n$.

\begin{figure}
 \begin{center}
  \includegraphics[width=7cm]{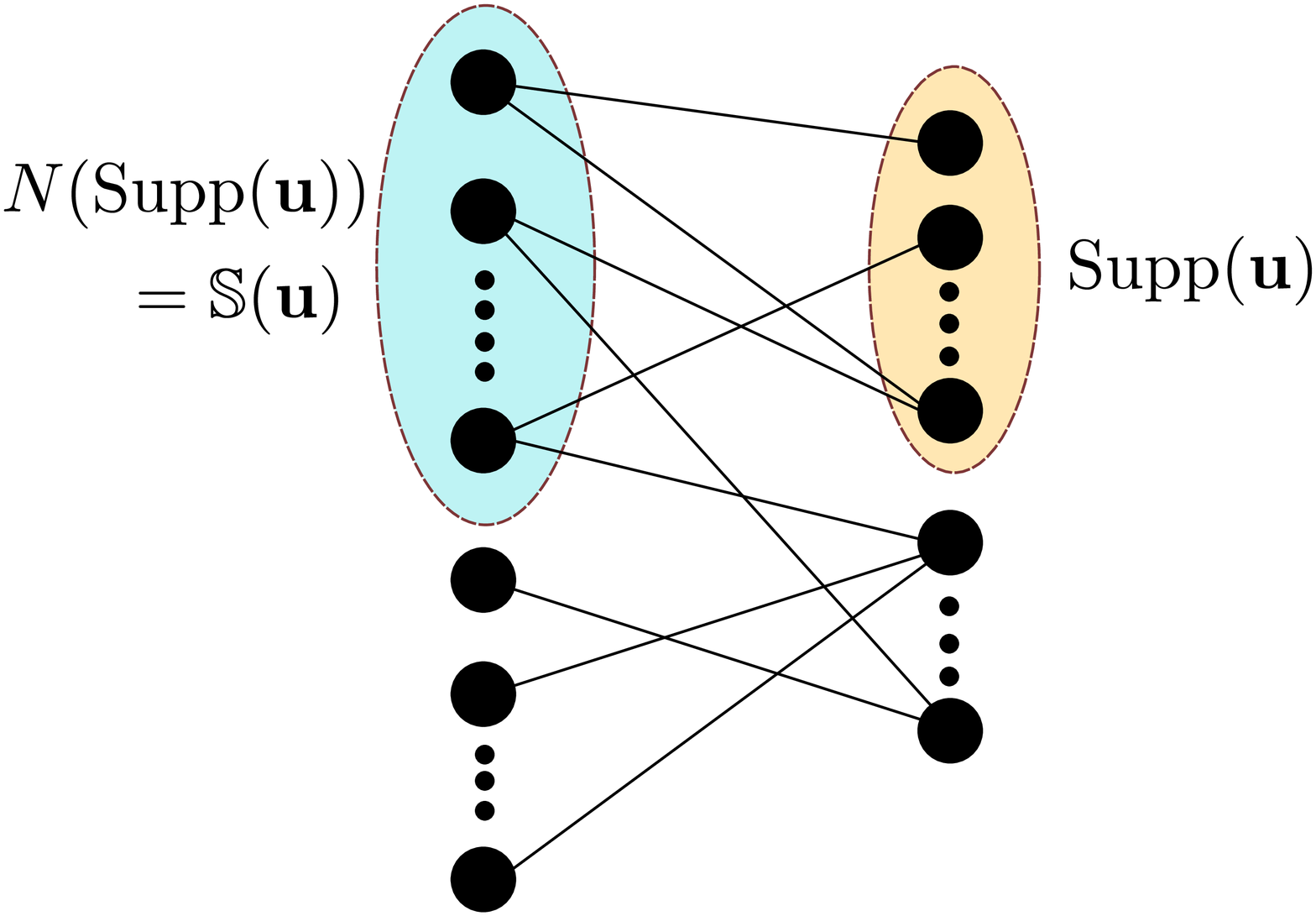}
\end{center}
\caption{Nodes corresponding to $\mathrm{Supp}(\u)$ and $\mathbb{S}(\u)$.}
\label{fig:Su_graph}
\end{figure}

For a given $\u\in \Fp^{n(1-R)}$, let us define $\mathbb{S}(\u)$ to be the set of all variable nodes that participate in the check equations $i$ for which the $i$th entry of $\u$ (i.e., $u_i$) is nonzero. 
Formally, $\mathbb{S}(\u):= \cup_{i\in \mathrm{Supp}(\u)}\mathrm{Supp}(\widehat{\h}_i)$. Equivalently, $i\in \mathbb{S}(\u)$ iff there exists 
$1\leq j\leq n(1-R)$ such that $u_j\neq 0$ and $\widehat{h}_{j,i}\neq 0$. This is illustrated in Fig.~\ref{fig:Su_graph}.

The rest of the article will be dedicated to proving the following theorem:
\begin{theorem}
 Let $A>2(1+R)$, $B>2(1+R)/(1-R)$, 
\[
 \epsilon=\frac{1-R}{A+1-R} \; \text{ and } \; \vartheta=\frac{1}{B(1-R)+1}.
\]
Suppose that $\Dta_V$ satisfies  (\ref{eq:Dta_V_condns}), and the corresponding $\cG$  is $(\alpha, A,\beta, B,\epsilon,\vartheta)$-good. Let
\begin{align}
 \lba &>\max \Bigg\{ \frac{1}{R},\frac{1}{1-R},\frac{2}{A-2(1-R)},\frac{2}{B(1-R)-2(1+R)},2\left( 1-\frac{1}{AB-1}-\frac{1}{A} \right)^{-1}, &\notag \\
                    &\hspace{6cm} \frac{1}{2(\alpha-1+R)} ,\frac{2B+3/2}{B(1-R)-1}  \Bigg\}. &\label{eq:lambda_conditions_main}
\end{align}
If we pick  $\L$ at random from the $(\cG,\lba)$ LDA ensemble, then the probability that $\L$ is simultaneously good for packing, channel coding, and MSE quantization 
tends to $1$ as $n\to\infty$. Moreover, the probability that $\L^*$ is also simultaneously good for packing, tends to $1$ as $n\to\infty$.
 \label{theorem:LDA_simultaneousgoodness}
\end{theorem}

We will prove each of the goodness properties in separate sections. The conditions on the parameters of the lattice to ensure goodness for channel coding are stated in Theorem~\ref{theorem:LDA_awgn}.
Goodness for packing is discussed in Corollary~\ref{theorem:LDA_packing}, and MSE quantization in Theorem~\ref{theorem:LDA_MSEquantization}.
Sufficient conditions for the packing goodness of the duals of LDA lattices are given in Theorem~\ref{theorem:LDA_dualpacking}.
The above theorem can then be obtained by a simple application of the union bound.
But before we proceed to the main results, we will discuss some useful lemmas that we will need later on in the proofs. 

\section{Some Preliminary Lemmas}\label{sec:prelim_lemmas}

In this section, we record some basic results that will be used in the proofs.
Recall that $V_n$ is the volume of a unit ball in $n$ dimensions.
We have the following upper bound on the number of integer points within a ball of radius $r$:
\begin{lemma}[Corollary of \cite{Ordentlich}, Lemma~1]
 Let $r>0$, $\y\in\R^n$, and $\cB$ denote the unit ball in $n$ dimensions. Then, 
     \tcol{
 \[
    V_n \left( \max\left\{0, r-\frac{\sqrt{n}}{2}\right\} \right)^n\leq|\Z^n\cap (\y+r\cB)|\leq  V_n \left( r+\frac{\sqrt{n}}{2} \right)^n.
 \]
   }
Furthermore, if $m\leq n$, then 
\[
     |\{ \x\in \Z^n\cap r\cB: |\mathrm{Supp}(\x)|\leq m \}|\leq \nchoosek{n}{m} V_m \left( r+\frac{\sqrt{m}}{2} \right)^m.
\]
\label{lemma:Zn_cap_rB}
\end{lemma}

Recall the randomized construction of the parity check matrix $H$ from the $(\alpha, A,\beta, B,\epsilon,\vartheta)$-good graph $\cG$, described in the previous section.
Also recall that for $\u\in \Fp^{n(1-R)}$, $\mathbb{S}(\u)$ is the set of all variable nodes that participate in the check equations $i$ for which  $u_i\neq 0$.
We have the following result which describes the distribution of $H^T\u$.
\begin{lemma}
 Let $\u\in \Fp^{n(1-R)}$, and $\x\in \Fp^n$. Then, 
 \begin{equation*}
  \Pr[H^T\u=\x]=\begin{cases}
                 \frac{1}{p^{|\Su|}} &\text{ if }\mathrm{Supp}(\x)\subset \Su \\
                 0 & \text{else.}
                \end{cases}
 \end{equation*}
\label{lemma:PrHu_x}
\end{lemma}
\begin{proof}
Let $\y:= H^T\u$.
The $j$th entry of $\y$ is given by $y_j=\sum_{i=1}^{n(1-R)}h_{ij}u_i$.
 Consider any $j\in (\Su)^c$. From the definition of $\Su$, it is easy to see that the $j$th variable node does not participate in any of the parity check equations indexed by $\mathrm{Supp}(\u)$.
 Hence, $h_{ij}=0$ whenever $u_i\neq 0$. 
 Therefore, $y_j=0$. On the other hand, if $j\in \Su$, then there exists at least one $i$ such that $h_{ij}\neq 0$. 
 So, $y_j=\sum_{i\in\mathrm{Supp}(\u)}h_{ij}u_i$, being a  nontrivial linear combination of independent and uniformly distributed random variables, is also uniformly distributed over $\Fp$. 
 Moreover, it is easy to see that the $y_j$'s are independent. Therefore, 
 \[
  \Pr[y_j=a]=\begin{cases}
              1/p & \text{ if } j\in \Su\\
              0  & \text { if } j\notin \Su \text{ and }a\neq 0\\
              1 & \text{ if } j\notin \Su \text{ and }a=0.
             \end{cases}
 \]
 
This completes the proof.
\end{proof}

Recall that $H$ defines a linear code over $\Fp$, where $p$ is the smallest prime greater than $n^\lba$.
 The following lemma, proved in Appendix~A, gives a lower bound on the probability of a randomly chosen $H$ not having full rank.
 \begin{lemma}
 If $B>2+(1+\dta)/\lba$ for some $\delta>0$, then 
\[
 \Pr[H \text{ is not full-rank }]\leq n^{-(2\lba+\dta)}(1+o(1)).
\]
\label{lemma:LDA_fullrank}
\end{lemma}
We now proceed to prove the various goodness properties of LDA lattices.

\section{Goodness for Channel Coding}\label{sec:channelcoding}

%
Recall that a sequence of  lattices $\{ \Ln \}$ is good for coding in presence of semi norm-ergodic noise if for any sequence of semi norm-ergodic noise vectors $\{\z^{(n)}\}$, with second moment per dimension
equal to $\sigma^2:= \frac{1}{n}\mathbb{E}[\Vert\z^{(n)}\Vert^2]$, 
\[
 \Pr[\z^{(n)}\notin \cV(\Ln)]\to 0 \text{ as }n\to\infty
\]
as long as $(\text{vol}(\Ln))^{2/n}>2\pi e \sigma^2$ for all sufficiently large $n$. But we have 
\[
(\text{vol}(\Ln))^{2/n}=(\reff(\Ln))^2V_n^{2/n}=(\reff(\Ln))^2\frac{2\pi e}{n}(1+o(1))
\]
using Stirling's approximation. Therefore, we can equivalently say that a sequence of lattices is good for coding in presence of semi norm-ergodic noise if $\Pr[\z^{(n)}\notin \cV(\Ln)]\to 0$ as $n\to\infty$
as long as $\reff(\Ln)\geq \sqrt{n\sigma^2}(1-o(1))$. 
Note that if the noise is assumed to be i.i.d.\  Gaussian, then the above definition is weaker than the definition of AWGN (or Poltyrev) goodness defined in~\cite{Erez05},
since the probability $\Pr[\z^{(n)}\notin \cV(\Ln)]$ is not required to go to zero exponentially in $n$. However, the above definition covers a much wider class
of noise distributions. In particular, the ``effective noise'' that is present in the equivalent modulo-lattice additive noise channel in the compute-and-forward protocol~\cite{Nazer11} is semi norm-ergodic, as discussed in~\cite{Ordentlich}.

The following result was proved by di Pietro:
\begin{theorem}[\cite{diPietrothesis}, Theorem~3.2]
 Let $\L$ be a lattice chosen uniformly at random from a $(\cG,\lba)$ LDA ensemble, where $\cG$   is $(\alpha, A,\beta, B,\epsilon,\vartheta)$-good, and (\ref{eq:Dta_V_condns}) is satisfied. If
 \[
  \lba>\max\left\{\frac{1}{2(\alpha-1+R)},\frac{3}{2(A-1+R)},\frac{1}{B(1-R)-1} \right\},
 \]
then the probability that $\L$ is good for coding in presence of semi norm-ergodic noise tends to 1 as $n\to\infty$.
\label{theorem:LDA_awgn}
\end{theorem}
For semi norm-ergodic noise $\{\z^{(n)}\}$, we have for every $\delta>0$, $\Pr[\z^{(n)}\notin (\sqrt{(1+\delta)n\sigma^2})\cB]\to 0$ as $n\to\infty$.
To prove that  $\{ \Ln\}$ is good for coding, it is then enough to show the absence of nonzero lattice points within a ball of radius $\sqrt{(1+\delta)n\sigma^2}$ around $\z$, for all $n\sigma^2<(\reff(\Ln))^2$
and all sufficiently large $n$.
In~\cite{diPietrothesis}, di Pietro proved the following statement, thus establishing Theorem~\ref{theorem:LDA_awgn}, and hence showing that LDA 
lattices are good for channel coding: For every $\mathbf{z}\in \sqrt{(1+\delta)n\nsvar}\cB$,
\begin{equation*}
  \sum_{\x\in \Z^n\cap (r_n \cB+\mathbf{z})\backslash p\Z^n}\Pr[\x\in\L]\to 0 \text{ as }n\to\infty,
  \label{eq:LDA_awgngood}
\end{equation*}
where $r_n=\reff(\Ln)(1+\delta_n)$, and $\delta_n\to 0$ as $n\to\infty$.

\section{Goodness for Packing}\label{sec:LDApacking}
Recall that  $\{\Ln\}$ is good for packing if
\[
 \limsup_{n\to\infty}\frac{\rpack(\Ln)}{\reff(\Ln)}\geq \frac{1}{2}.
\]
The packing goodness of LDA lattices follows as a corollary to Theorem~\ref{theorem:LDA_awgn}.
\begin{corollary}
 Let $\L$ be a lattice chosen uniformly at random from a $(\cG,\lba)$ LDA ensemble,  where $\cG$   is $(\alpha, A,\beta, B,\epsilon,\vartheta)$-good, and (\ref{eq:Dta_V_condns}) is satisfied.
 Furthermore, let 
 \[
  \lba>\max\left\{\frac{1}{2(\alpha-1+R)},\frac{3}{2(A-1+R)},\frac{1}{B(1-R)-1} \right\}.
 \]
Then, the probability that $\L$ is good for packing tends to 1 as $n\to\infty$.
\label{theorem:LDA_packing}
\end{corollary}
\begin{proof}
Let us choose $r_n=\reff(\L)(1-\delta_n)$, where $\delta_n$ is a quantity that goes to $0 $ as $n\to\infty$.
We want to prove that 
\[
 \Pr[\rpack(\L)< r_n/2]\to 0 \text{ as } n\to\infty.
\]
It is enough to show that the probability of any nonzero integer point within $r_n\cB$ belonging to  $\L$ goes to zero as $n\to\infty$, i.e., 
\[
  \sum_{\x\in \Z^n\cap r_n \cB\backslash\{ \0\}}\Pr[\x\in\L]\to 0 \text{ as }n\to\infty
\]
This requirement is similar to (\ref{eq:LDA_awgngood}), and the rest of the proof of packing goodness of LDA lattices follows, \emph{mutatis mutandis}, on similar lines as that for goodness for channel coding.
\end{proof}
\section{Goodness for MSE Quantization}\label{sec:LDA_MSEquantization}

In nested lattice coding for power constrained transmission over Gaussian channels, the codebook is generally  the set of all
points of the fine lattice within the fundamental Voronoi region of the coarse lattice.
Hence, the fine lattice determines the codeword points, while the coarse lattice defines the shaping region.
In order to maximize the rate for a given power constraint, we want the shaping region to be approximately spherical.
The loss in rate (penalty for not using a spherical shaping region) is captured by the normalized second moment, $G(\L)$, of the coarse lattice $\L$, and in order to minimize 
this loss, we want $G(\L)$ to be as close to the asymptotic normalized second moment of a sphere as possible.
As defined in Section~\ref{sec:basicdefs},  $\{\Ln\}$ is  good for MSE quantization if ${G}(\Ln)\to\frac{1}{2\pi e}$ as $n\to\infty$.
In this section, we will prove the following result:
\begin{theorem}
 Let $A>2(1+R)$ and $B>2(1+R)/(1-R)$. Fix
\[
 \epsilon=\frac{1-R}{A+1-R} \; \text{ and } \; \vartheta=\frac{1}{B(1-R)+1}.
\]
Suppose that $\Dta_V$ satisfies the conditions of Lemma~\ref{lemma:goodgraphs}, and $\cG$ is $(\alpha, A,\beta, B,\epsilon,\vartheta)$-good. Furthermore, let
\begin{equation}
 \lba>\max \left\{ \frac{1}{R},\frac{1}{1-R},\frac{2}{A-2(1+R)},\frac{2}{B(1-R)-2(1+R)},2\left( 1-\frac{1}{AB-1}-\frac{1}{A} \right)^{-1} \right\}.
\label{eq:lambda_conditions}
\end{equation}
Let $\L$ be randomly chosen from a $(\cG,\lba)$ LDA ensemble.
Then, the probability that $\L$ is good for MSE quantization tends to $1$ as $n\to\infty$.
\label{theorem:LDA_MSEquantization}
\end{theorem}
To prove the theorem, we will show that for every positive $\delta_1,\delta_2$, and all sufficiently large $n$,
\begin{equation}
  \Pr\left[G(\L)>\frac{1}{2\pi e}+\delta_1\right]\leq \delta_2.
\label{eq:MSEquantization_main}
\end{equation}
Since $G(\L)>1/(2\pi e)$ for all $\L$~\cite{Erez05}, the above statement guarantees the existence of a sequence of lattices, $\{\Ln\}$, for which $G(\Ln)\to 1/(2\pi e)$ as $n\to\infty$.
 Our proof of the above inequality is based on the techniques used in~\cite{Ordentlich} and~\cite{diPietrothesis}.
 For a lattice $\L$, and $\x\in \R^n$, we define $d(\x,\L):=\min_{\y\in \L}\Vert \x-\y\Vert$ to be the Euclidean distance between $\x$ and the closest point in $\L$ to $\x$.
 For ease of notation, let us define $r:=\reff(\L)$.
Our proof of inequality (\ref{eq:MSEquantization_main}), and hence Theorem~\ref{theorem:LDA_MSEquantization}, will make use of the following lemmas, which are 
proved in Appendix~B.

\begin{lemma}
 Suppose that the hypotheses of Theorem~\ref{theorem:LDA_MSEquantization} are satisfied. Let $\L$ be drawn uniformly at random from a $(\cG,\lba)$ LDA ensemble, 
 and $X$ be a random vector uniformly distributed over $\cV(\L)$. 
 Then, 
 \begin{equation}
  \bE_{\L}[G(\L)]\leq \bE_{\L,X}\left[\frac{d^{2}(X,\L)}{n(\textnormal{vol}(\L))^{2/n}}\Bigg| H\text{ is full rank}\right] +o(1). 
 \label{eq:MSEgood_lemma1}
 \end{equation}
\label{lemma:MSEgood_1}
\end{lemma}


\begin{lemma}
 Suppose that the hypotheses of Theorem~\ref{theorem:LDA_MSEquantization} are satisfied. Let $0<\omega<1$. There exists a $\dta>0$ so that for every $\x\in\R^n$,
 \begin{equation}
 \Pr\left[ d(\x,\L)>r\left( 1+\frac{1}{n^\omega} \right)\Bigg| H\text{ is full rank} \right] \leq \frac{1}{n^{2\lba R+\dta}}(1+o(1)).
\label{eq:dxL_bound}
\end{equation}
\label{lemma:dxL_bound}
\end{lemma}
\begin{lemma}
 Let $U$ be a random vector uniformly distributed over $[0,p)^n$, and $X$ be uniformly distributed over $\cV(\L)$. Then,
  \begin{equation}
   \mathbb{E}_{\L}\mathbb{E}_{X}[d^2(X,\L)|H\text{ is full rank}]=\mathbb{E}_{U}\mathbb{E}_{\L}[d^2(U,\L)|H\text{ is full rank}].
  \end{equation}
\label{lemma:ELX_EUL}
\end{lemma}
\begin{proof}[Proof of Theorem~\ref{theorem:LDA_MSEquantization}]

Recall that to prove the theorem, it is enough to prove inequality (\ref{eq:MSEquantization_main}). To this end, we will show that the first term in (\ref{eq:MSEgood_lemma1})
tends to $1/(2\pi e)$ as $n\to\infty$.
We will use Lemma~\ref{lemma:dxL_bound} to bound this term.

Recall that $r=\reff(\L)$.
Since (\ref{eq:dxL_bound}) holds for all $\x\in \R^n$, we can say that for any random vector $U$ (having density function $f$) over $\R^n$,  we have
\begin{align}
  \Pr\big[d(U,\L) >r(1+n^{-\omega})\big|H\text{ is full rank}\big] &=\int_{\R^n} \Pr\big[d(\u,\L)>r(1+n^{-\omega})\big|H\text{ is full rank}\big]f(\u)d\u & \notag \\
                                                                                            &\leq n^{-(2\lba R+\dta)}(1+o(1)). &\notag
\end{align}
Let us define $\rho =r(1+n^{-\omega})$.
For any $\u\in \R^n$, and any Construction-A lattice $\L$, we have $d(\u,\L)\leq p\sqrt{n}/2$.
Then, for any distribution on $U$,
\begin{align}
  \mathbb{E}_{U}\mathbb{E}_{\L}\big[d^2(U,\L)\big|H\text{ is full rank}\big] &\leq \rho^2 \Pr\big[d(U,\L)\leq \rho\big|H\text{ is full rank}\big] &\notag \\
                                                                                                     &\qquad\qquad +\frac{p^2n}{4}\Pr\big[d(U,\L)>\rho\big|H\text{ is full rank}\big] &\notag \\
                                                                                    &\leq \rho^2\left( 1+ \frac{p^2 n}{4\rho^2}\frac{1}{n^{2\lba R+\dta}}(1+o(1))\right) .&\notag
\end{align}
Substituting  $\rho=\frac{n^{\lba (1-R)+1/2}}{\sqrt{2\pi e}}(1+o(1))$,
\begin{align}
  \mathbb{E}_{U}\mathbb{E}_{\L}\big[d^2(U,\L)\big|H\text{ is full rank}\big]  &\leq \rho^2 \left( 1+n^{2\lba+1}  \frac{2\pi e}{4n^{2\lba(1-R)+1}}\frac{1}{n^{2\lba R+\dta}}(1+o(1)) \right) &\notag\\
                                                                                    &= \rho^2 \left(1+ \frac{\pi e}{2n^{\dta}}(1+o(1))\right)&\notag \\
                                                                                    &= r^2 (1+o(1)). &\label{eq:EUL_bound}
\end{align}

From (\ref{eq:EUL_bound}) and Lemma~\ref{lemma:ELX_EUL}, we have 
\[
     \mathbb{E}_{\L}\mathbb{E}_{X}\big[d^2(U,\L)\big|H\text{ is full rank}\big]\leq r^2 (1+o(1)).
\]
Recall that $V_n$ denotes the volume of an $n$-dimensional unit ball. 
Using Stirling's approximation, we get,
\[
 V_n^{1/n}=\left( \frac{\pi^{n/2}}{\Gamma(n/2+1)} \right)^{1/n}= \frac{\sqrt{2\pi e}}{n^{1/2}}(1+o(1)).
\]
 Therefore,
\[
 n(\text{vol}(\L))^{2/n}=(\reff(\L))^2 2\pi e(1+o(1))=r^2 2\pi e(1+o(1))
\]
 and hence,
 \[
   \mathbb{E}_{\L}\mathbb{E}_{X}\left[\frac{d^2(U,\L)}{n(\text{vol}(\L))^{2/n}}\Bigg|H\text{ is full rank}\right]\leq \frac{1}{2\pi e}(1+o(1)).
 \]
 Using this, and Lemma~\ref{lemma:MSEgood_1},
we can write
\begin{equation}
  \mathbb{E}[G(\L)] \leq \frac{1}{2\pi e} (1+\delta(n)), 
  \label{eq:msegood_1}
\end{equation}
where $\delta(n)$ is a quantity that goes to $0$ as $n\to \infty$.
We also have $G(\L)>1/(2\pi e)$ for all $\L$. For any $\gamma>0$, we can write
\begin{align}
   \mathbb{E}[G(\L)]&\geq \frac{1}{2\pi e}\Pr\left[\frac{1}{2\pi e}<G(\L)\leq \frac{1}{2\pi e}+\gamma\right] +\left(\frac{1}{2\pi e}+\gamma\right)\Pr\left[G(\L)> \frac{1}{2\pi e}+\gamma\right]&\notag \\
                                 &= \frac{1}{2\pi e}\left( 1-\Pr\left[G(\L)> \frac{1}{2\pi e}+\gamma\right]\right) +\left(\frac{1}{2\pi e}+\gamma\right)\Pr\left[G(\L)> \frac{1}{2\pi e}+\gamma\right]&\notag \\
                                 &=\frac{1}{2\pi e}+ \gamma \Pr\left[G(\L)> \frac{1}{2\pi e}+\gamma\right], &\notag
\end{align}
and hence,
\begin{equation*}
  \Pr\left[G(\L)> \frac{1}{2\pi e}+\gamma\right]\leq \frac{\mathbb{E}[G(\L)]-1/(2\pi e)}{\gamma}
\end{equation*}
Since the above inequality holds for every $\gamma>0$, we can choose, for e.g., $\gamma=\sqrt{\delta(n)}$, and use (\ref{eq:msegood_1}) to obtain
\[
  \Pr\left[G(\L)>\frac{1}{2\pi e}+\sqrt{\delta(n)}\right]\leq \sqrt{\delta(n)}\to 0 \text{ as } n\to\infty.
\]
Therefore, we can conclude that the probability of choosing an LDA lattice which is good for MSE quantization tends to $1$ as $n\to\infty$. 
\end{proof}

\section{Packing Goodness of the Duals of LDA Lattices}\label{sec:LDA_dualpacking}
Recall that $\rpack(\L)$ denotes the packing radius of $\L$, and that
a sequence of lattices $\{ \Ln \}$ is good for packing if 
\[
 \frac{\rpack(\Ln)}{\reff(\Ln)}\geq \frac{1}{2} -o(1).
\]

Our motivation for studying the properties of the dual of a lattice comes from~\cite{Vatedka13}, where a nested lattice coding scheme was presented for compute-and-forward in a bidirectional relay network with an untrusted relay. In this problem, two users want to exchange messages with each other, with all communication taking place via an honest-but-curious bidirectional relay.
The users operate under an average transmission power constraint of $P$, and the links between the users and the relay are AWGN channels with noise variance $\nsvar$. 
The messages have to be  reliably exchanged (the probability of decoding error should go to zero asymptotically in the blocklength), but kept secret from the relay. To be more specific,
the signals received by the relay have to be statistically independent of the individual messages. This requirement is also called perfect (or Shannon) secrecy.
It was shown in~\cite{Vatedka13} that if the fine lattices are good for AWGN channel coding, the coarse lattices are good for MSE quantization, and the duals of the coarse lattices
are good for packing, then a rate of $\frac{1}{2}\log_2\frac{P}{\nsvar}-\log_2(2e)$ can be achieved with perfect secrecy.
This motivates us to construct lattices whose duals are good for packing. In this section, we will prove the following result.
\begin{theorem}
 Let $\cG$ be an $(\alpha, A,\beta, B,\epsilon,\vartheta)$-good $(\Dta_V,\Dta_C)$-regular bipartite graph whose parameters satisfy the hypotheses of Lemma~\ref{lemma:goodgraphs}. If 
 \[
     \lba>\max\left\{\frac{1}{2(1-R)}, \:\frac{2B+3/2}{B(1-R)-1}\right\},
 \]
 then the dual of a randomly chosen lattice from a $(\cG,\lba)$ LDA ensemble is good for packing with probability tending to $1$ as $n\to\infty$.
 \label{theorem:LDA_dualpacking}
\end{theorem}

\begin{proof}
If $\L$ is a lattice obtained by applying Construction A to a linear code $\cC$, and if $\L^*$ is the dual of $\L$, then, $\frac{1}{p}\L^*$ is obtained 
 by applying Construction A to the dual code, $\cC^{\perp}$ (see~\cite[Lemma~27]{Vatedka13} for a proof). 
 To show that the duals of LDA lattices are good for packing, it is enough to show that the Construction-A lattices generated by the duals
 of the nonbinary LDPC codes ($\cC$) are good for packing.
 
 Note that $H$ (a parity check matrix for $\cC$) is a generator matrix for $\cC^{\perp}$. Let $\L'$ be the lattice obtained by applying Construction A on $\cC^{\perp}$.
 We will prove that $\L'$ is good for packing.
 The lattice $\L'$ contains $p\Z^n$ as a sublattice, and the nesting
 ratio is $p^{n(1-R)}$ if $H$ is full-rank. The volume of $\cV(\L')$ is equal to the ratio of the volume of $p\Z^n$ to the nesting ratio, and hence,
 \[
     \vol(\L')=\frac{p^n}{p^{n(1-R)}}=p^{nR}.
 \]
Recall that $V_n$ is the volume of the unit ball in $n$ dimensions. The effective radius of $\L'$ can therefore be written as,
\begin{equation}
     \reff(\L')=\frac{p^R}{(V_n)^{1/n}}.
     \label{eq:reff_dual}
\end{equation}
Let us define
\begin{equation*}
 r_n:= \frac{p^R}{V_n^{1/n}}\zeta_n,
\label{eq:rn_defn}
 \end{equation*}
where $\zeta_n$ is a term that goes to $1$ as $n\to\infty$, defined as follows:
\begin{equation*}
 \zeta_n=\frac{1}{n^{4/n}}\left(\frac{ C_1}{e(1-R)\ln n} \right)^{\frac{4C_1}{(1-R)\ln n}} \left(1-\frac{C_1}{(1-R)\ln n}\right)^{2}.
\label{eq:zetan_defn}
 \end{equation*}
Here,
\begin{equation}
 C_1:=\frac{\ln\left( \frac{8}{1-(1-R)/(2\alpha)} \right)}{\lba(1-(1-R)/\alpha)}.
\label{eq:C1_defn}
 \end{equation}

 We want to prove that the probability $\Pr[\rpack(\L')<\reff(\L')/2]\to 0$ as $n\to\infty$.
 We will show that the probability of finding a nonzero lattice point within a ball of radius $r_n$ centered at $\0$ goes to zero as $n\to\infty$.
 
Since $p\Z^n$ is always a sublattice of $\L'$, we must ensure that $\reff(\L')<p$. Substituting for $\reff(\L')$ from (\ref{eq:reff_dual}), we can see that $\reff(\L')<p$ is satisfied
 for all sufficiently large $n$ as long as $\lba>\frac{1}{2(1-R)}$, which is guaranteed by the hypothesis of Theorem~\ref{theorem:LDA_dualpacking}. 

 We want
 \[
    \Pr\left[ \exists \u\in \Fp^{n(1-R)}\backslash\{ \0 \}:H^T\u  \in (\Z^n\cap r_n\cB)\bmod p\Z^n \right]\to 0 \text{ as }n\to\infty.
 \]
Instead, we will prove the following (stronger) statement.
 \[
     \sum_{\u\in\Fp^{n(1-R)}\backslash\{  \0\}}\Pr\left[ H^T\u \in (\Z^n\cap r_n\cB)\bmod p\Z^n \right]\to 0 \text{ as }n\to\infty.
 \]
 The summation in the above statement can be expanded as follows:
 \begin{align}
   &\sum_{\u\in\Fp^{n(1-R)}\backslash\{  \0\}}\Pr\left[ H^T\u  \in (\Z^n\cap r_n\cB)\bmod p\Z^n \right]&\notag \\
   &\quad =\sum_{t=1}^{n(1-R)}\sum_{\substack{\u\in\Fp^{n(1-R)}\\ |\mathrm{Supp}(\u)|=t}}\Pr\left[ H^T\u  \in (\Z^n\cap r_n\cB)\bmod p\Z^n \right]&\notag \\
   &\quad =\sum_{t=1}^{n(1-R)}\sum_{\substack{\u\in\Fp^{n(1-R)}\\ |\mathrm{Supp}(\u)|=t}}\sum_{s=1}^{n}\sum_{\substack{\x\in\Z^n\cap r_n\cB \\ |\mathrm{Supp}(\x)|=s}}\Pr[H^T\u \equiv \x\bmod p]. &\label{eq:prpack_1}
 \end{align}

Fix $\u\in \Fp^{n(1-R)}$. Recall, from Section~\ref{sec:LDA_ensemble}, that $\mathbb{S}(\u)$ is the set of all variable nodes that participate in the check equations $i$ for which $u_i\neq 0$. 
  For $S\subset\{ 1,2,\ldots,n \}$, define $\mathbf{1}_S(\Su)$ to be the function that takes the value  $1$ if $\Su=S$, and zero otherwise. 
  Note that this is a deterministic function of $\u$ since $\hH$ is fixed beforehand.
  Let us also define $\mathtt{1}_m(\Su)$ to be the function which takes the value $1$ if $|\Su|=m$, and zero otherwise.
  Using Lemma~\ref{lemma:PrHu_x}, we have
  \begin{align}
   \Pr[H^T \u \equiv \x\bmod p] &= \begin{cases}
								 \frac{1}{p^{|\Su|}} & \text{if }\Supp(\x)\subset \Su \\
								 0                                 &\text{otherwise.}
                                                                \end{cases}  &\notag \\
                                                       &\leq \frac{1}{p^{|\Su|}} &\notag \\
                                                       &= \sum_{m=1}^{n}\mathtt{1}_m(\Su) \frac{1}{p^m}. &\notag
  \end{align}

We use this in (\ref{eq:prpack_1}) to obtain
\begin{align}
  &\sum_{\u\in\Fp^{n(1-R)}\backslash\{  \0\}}\Pr\left[ H^T\u \in (\Z^n\cap r_n\cB)\bmod p\Z^n \right] &\notag \\
  &\quad \leq\sum_{t=1}^{n(1-R)}\sum_{\substack{\u\in\Fp^{n(1-R)}\\ |\mathrm{Supp}(\u)|=t}} \sum_{m=1}^{n}\mathtt{1}_m(\Su)\frac{1}{p^m}
                                                      \sum_{s=1}^{m} \sum_{\substack{\x\in\Z^n\cap r_n\cB \\ |\mathrm{Supp}(\x)|=s}}1&\notag \\
  &\quad =\sum_{t=1}^{n(1-R)}\sum_{\substack{\u\in\Fp^{n(1-R)}\\ |\mathrm{Supp}(\u)|=t}} \sum_{m=1}^{n}\mathtt{1}_m(\Su)
                                                     \frac{1}{p^m}\vert\{\x\in\Z^n\cap r_n\cB : |\mathrm{Supp}(\x)|\leq m\}\vert .&\label{eq:LDAdualpacking_step1}
\end{align}
In Appendix~C, we show that the above quantity goes to zero as $n\to\infty$. Therefore, the probability that the dual of a randomly chosen LDA lattice
is good for packing goes to 1 as $n\to\infty$, completing the proof of Theorem~\ref{theorem:LDA_dualpacking}.
\end{proof}
\section{Remarks}\label{sec:remarks}

We now make some observations regarding our results and their applications to several problems. 
We first discuss the extension of our results to nested lattices, and then make some remarks regarding the choice of parameters,
before concluding with some open problems.

\subsection{Construction of Nested Lattices}

The main result of our paper, namely Theorem~\ref{theorem:LDA_simultaneousgoodness}, shows that a randomly chosen LDA lattice satisfies the desired ``goodness'' properties with high probability.
In applications such as compute-and-forward, and coding for the power constrained AWGN channel, we need \emph{nested lattices} which satisfy
the necessary properties. Different nested lattice constructions have been proposed~\cite{Erez04,Nazer11,Ordentlich}, and 
we briefly describe the construction by Ordentlich and Erez~\cite{Ordentlich} here, since the results presented in this paper
can be easily extended to nested lattices using their construction.

Choose a $k_c\times n$ parity check matrix, $H_c$, over $\Fp$. Let $\cC_c$ be the linear code that has parity check matrix $H_c$.
Let $H_f$ be the $k_f\times n$ parity check matrix ($k_f<k_c$) that consists of the first $k_f$ rows of $H_c$, and $\cC_f$ denote
the corresponding linear code. Clearly, $\cC_c$ is a subcode of $\cC_f$. If $\L_c$ and $\L_f$ are lattices obtained by applying 
Construction A to $\cC_c$ and $\cC_f$ respectively, then $\L_c\subset \L_f$, with nesting ratio $p^{k_f-k_c}$ if the rows of $H_c$ are linearly independent. 
The parity check matrix $H_c$ can be chosen so that 
the Tanner graphs corresponding to both $\cC_c$ and $\cC_f$ have the required expansion properties~\cite[Section 4.3]{diPietrothesis}. As long as $\lba$ and the 
parameters of the Tanner graph are chosen appropriately, the lattice $\L_c$ satisfies the goodness properties with 
probability tending to $1$ as $n\to\infty$. Also, $\L_f$ satisfies the goodness properties with high probability. Using the union bound,
we can argue that $\L_c$ and $\L_f$ simultaneously satisfy the goodness properties with probability tending to $1$ as $n\to\infty$.

With this construction, we can use Theorem~\ref{theorem:LDA_simultaneousgoodness} to conclude that nested LDA lattices achieve the capacity of the power constrained AWGN channel, 
the capacity of the dirty paper channel, and the rates guaranteed by
the compute-and-forward protocol~\cite{Nazer11}. Furthermore, they can also be used for secure bidirectional relaying,
and achieve the rates guaranteed by~\cite{Vatedka13}.  However, all of this is guaranteed under the assumption of
a \emph{closest lattice-point decoder} being used at the destination/relay. Although these lattices were empirically shown to give
low error probability over the AWGN channel (without power constraints), their performance with
belief propagation decoding still requires further study. 

\subsection{Choice of Parameters and Complexity of the BP Decoder}
Theorem~\ref{theorem:LDA_simultaneousgoodness} gives sufficient conditions on the parameters required to obtain 
the structural goodness properties of a randomly chosen LDA lattice.  In practice, one would want to optimize over 
the parameters in Theorem~\ref{theorem:LDA_simultaneousgoodness} to reduce the decoding complexity. 
At this point, we can only say that the achievability results for the various communication problems
are valid with the assumption that a closest lattice-point decoder is used.
However, in practice, we would want to use a belief propagation decoder instead.
If this is done,
then the decoding complexity would be roughly of the order of $np\log p$ ($p$ messages need to be computed at each node,
 this having complexity $O(p\log p)$, and there are $O(n)$ nodes). Therefore, it is necessary to choose the smallest $p$
for which the conditions of Theorem~\ref{theorem:LDA_simultaneousgoodness} are satisfied. Note that the condition $\lba>\frac{2B+3/2}{B(1-R)-1}$
means that we should always have $\lba>2/(1-R)$.  Choosing $R=1/3$, we can
make the lower bound on $\lba$ close to $3$ by appropriately choosing $A$ and $B$. 
This means that the decoding complexity would be roughly of the order of $n^4\log n$.
Although this means that we can decode in polynomial time, this complexity is still high when compared to the decoding complexity of the lattices presented in~\cite{Sommer,Yan13}. 
  \tcol{
For practical implementation of LDA lattices, it would be desirable to have the decoding complexity grow as  $O(n)$ or $O(n\log^{\gamma}n)$ for some $\gamma>0$.
For instance, the encoding and decoding complexities of polar lattices~\cite{Yan13} grow as $O(n\log^2n)$.
  }
However, it is still not known whether the lattices in~\cite{Sommer,Yan13} have all the ``goodness'' properties that the LDA lattices satisfy. 

\subsection{Some Future Directions}\label{sec:conclusions}

As remarked earlier, the study of BP decoders for LDA lattices requires further investigation, and empirical evidence suggests that
LDA lattices perform well with BP decoding.
Another key point to note is that we required  $p$ to grow polynomially in $n$ to obtain the aforementioned goodness properties.
Large values of $p$ translate to higher BP decoding complexity, and it would be useful to study the 
structural properties of LDA lattices over fields of smaller sizes. Empirical results by~\cite{diPietrothesis,lda_ita13}
suggest that it may be possible to get good error performance over the AWGN channel (without power constraints) even with moderate field sizes. This suggests that it may be possible to tighten the arguments presented in this paper and obtain better bounds on the parameters of the LDA lattices needed to guarantee the desired properties. 

In this article, we did not discuss two important ``goodness'' properties, namely covering goodness,
and secrecy goodness~\cite{Ling13} of LDA lattices. The property of secrecy goodness was
crucially used in designing nested lattice codes for the wiretap channel in~\cite{Ling13},
and for strongly secure bidirectional relaying in~\cite{Vatedka13}.
Whether LDA lattices satisfy these properties is left as future work.

\section{Acknowledgements}
The authors would like to thank Gilles Z\'{e}mor for useful discussions. The work of the first author was supported in part by the Tata Consultancy Services Research Scholarship Program.
\section*{Appendix A: Proof of Lemma~\ref{lemma:LDA_fullrank}}

We will prove that the probability that there is any nontrivial  linear combination of the rows of $H$ equal to zero tends to $0$ as $n\to\infty$.
Let $\h_i$ denote the $i$th row of $H$.
For any $S\subseteq\{ 1,2,\ldots,n(1-R) \}$, we define 
\[
 \chi_S = \begin{cases}
             1 & \text{if there exist }\{ a_i:i\in S,a_i\in \Fp\backslash\{0\} \} \text{ such that } \sum_{i\in S}a_i\mathbf{h}_i=\0 \\
             0 & \text{otherwise.}
          \end{cases}
\]
Let us also define 
\[
  Y = \sum_{s=1}^{n(1-R)}\sum_{\substack{S\subset\{ 1,2,\ldots,n(1-R) \}\\ |S|=s}}\chi_S
\]
Clearly, $H$ is full rank if and only if $Y=0$.
Using Markov's inequality, we see that 
\[
 \Pr[Y\geq 1]\leq \mathbb{E}[Y].
\]
Therefore, it is enough to find an upper bound on the expectation of $Y$. Let 
\[
\eta(S) = |\cup_{i\in S}\mathrm{Supp}(\h_i)|.
\]
In other words, $\eta(S)$ is the number of variable nodes that participate in the parity check equations indexed by $S$. This is also equal to the number of neighbours of  $S$  in $\cG$, i.e., $|N(S)|$.
Observe that there are at most $p^{s-1}$ different linear combinations (not counting scalar multiples) of $s$ rows of $H$. Using Lemma~\ref{lemma:PrHu_x}, the probability that a fixed linear combination of the $S$ rows of $H$ is zero is equal to $1/p^{\eta(S)}$. Using the union bound,
\[
 \Pr[\chi_S=1]\leq \frac{p^{s-1}}{p^{\eta(S)}}.
\]
Therefore, we have
\begin{align}
 \bE[Y] &\leq \sum_{s=1}^{n(1-R)} \sum_{\substack{S\subset\{ 1,2,\ldots,n(1-R) \}\\ |S|=s}}\frac{p^{s-1}}{p^{\eta(S)}} &\notag \\
           & = \sum_{s=1}^{\vartheta n(1-R)} \sum_{\substack{S\subset\{ 1,2,\ldots,n(1-R) \}\\ |S|=s}}\frac{p^{s-1}}{p^{\eta(S)}}+\sum_{s=\vartheta n(1-R)}^{n(1-R)/2} \sum_{\substack{S\subset\{ 1,2,\ldots,n(1-R) \}\\ |S|=s}}\frac{p^{s-1}}{p^{\eta(S)}} &\notag \\
              &  \qquad\qquad\qquad + \sum_{s=n(1-R)/2}^{n(1-R)} \sum_{\substack{S\subset\{ 1,2,\ldots,n(1-R) \}\\ |S|=s}}\frac{p^{s-1}}{p^{\eta(S)}} &\notag 
\end{align}
It is easy to show that for every $S\subset C$, we have $\eta(S)=|N(S)|\geq |S|/(1-R)$.
Using properties (R1) and (R2) of the superexpanders, and the above fact, we have
\begin{align}
   \bE[Y]       & \leq  \sum_{s=1}^{\vartheta n(1-R)}\nchoosek{n(1-R)}{s}\frac{p^{s-1}}{p^{Bs}} +  \sum_{s=\vartheta n(1-R)}^{n(1-R)/2}\nchoosek{n(1-R)}{s}\frac{p^{s-1}}{p^{\beta s}} &\notag \\
         & \qquad \qquad\qquad   +\sum_{s=n(1-R)/2}^{n(1-R)}\nchoosek{n(1-R)}{s} \frac{p^{s-1}}{p^{s/(1-R)}}, &\notag
\end{align}
 We can further simplify this as follows,
\begin{align}
 \bE[Y]   &\leq  \sum_{s=1}^{\vartheta n(1-R)} n^s \frac{n^{\lba (s-1)}}{n^{\lba Bs}} + \sum_{s=\vartheta n(1-R)}^{n(1-R)/2}2^n\frac{n^{\lba(s-1)}}{n^{\lba\beta s}}  +\sum_{s=n(1-R)/2}^{n(1-R)}2^n \frac{n^{\lba(s-1)}}{n^{\lba s/(1-R)}}  & \notag \\
       &\leq \sum_{s=1}^{\vartheta n(1-R)} n^{s(1+\lba (1-B))-\lba} + n^{-c_1 n}(1+o(1))  &\notag \\
      &= n^{(1+\lba (1-B))-\lba}(1+o(1))+ n^{-c_1 n}(1+o(1)),&\label{eq:EX_upperbound}  
\end{align}
for some constant $c_1>0$, since $\beta$ and $1/(1-R)$ are greater than $1$, and $B>1+1/\lba$. Suppose that for some constant $\dta>0$, we have $B>2+(1+\dta)/\lba$. Then, $(1+\lba (1-B))-\lba<-(2\lba+\dta)$, and therefore,
\[
 \bE[Y] \leq n^{-(2\lba+\dta)}(1+o(1)).
\]
Therefore, $\Pr[Y\geq 1]$, and hence the probability that $H$ is not full rank, goes to zero as $n\to\infty$.\qed

{\bf{Remark:}}
To prove that $\bE[Y]\to 0$ in (\ref{eq:EX_upperbound}), it is sufficient to have $B>1+1/\lba$. 
The expected value of $Y$, and subsequently $\Pr[H \text{ is not full-rank}]$ could then be bounded from above by $n^{-(\lba+\dta)}(1+o(1))$.
However, we need $\Pr[H \text{ is not full-rank}]$ to be less than $n^{-(2\lba+\dta)}(1+o(1))$ to prove that LDA lattices are good for MSE quantization (in particular, to show that the second term in (\ref{eq:EGl_ubound}) goes to zero),
and hence we impose the stronger condition that $B>2+1/\lba$.
\section*{Appendix~B}

\subsection{Proof of Lemma~\ref{lemma:MSEgood_1}}
Recall that $V_n$ denotes the volume of an $n$-dimensional unit ball. Using Stirling's approximation, we get,
\begin{equation*}
 V_n^{1/n}=\left( \frac{\pi^{n/2}}{\Gamma(n/2+1)} \right)^{1/n}= \frac{\sqrt{2\pi e}}{n^{1/2}}(1+o(1)).
\label{eq:Vn_Stirlingapprox}
\end{equation*}
For any Construction-A lattice $\L$, we have $p\Z^n\subset \L$. If $H$ is full-rank, 
then the number of points of $\L$ in $[0,p)^n$, (and therefore, within $\cV(p\Z^n)$) is equal to $p^{nR}$, which is $|\L/p\Z^n|$.
Since $|\L/p\Z^n|=\text{vol}(p\Z^n)/\text{vol}(\L)$,
we get $\text{vol}(\L)=p^{n(1-R)}$. Therefore,
\begin{equation*}
 \reff(\L)=\left( \frac{\text{vol}(\L)}{V_n} \right)^{1/n}=\frac{n^{\lba (1-R)+1/2}}{\sqrt{2\pi e}}(1+o(1)).
\end{equation*}

For any $\x\in\R^n$, we have
$d(\x,\L)=\min_{\y\in\L}\Vert \y-\x \Vert$ to be the distance between $\x$ and the closest point in $\L$ to $\x$.
Recall that $X$ is a random vector uniformly distributed over the fundamental Voronoi region of $\L$.
The normalized second moment of $\L$ is then equal to 
\[
 G(\L)=\mathbb{E}_X\left[ \frac{d^2(X,\L)}{n(\text{vol}(\L))^{2/n}} \right].
\]
We can write
\begin{align}
 \bE_{\L}[G(\L)] &= \bE_{\L}[G(\L)| H\text{ is full rank}]\Pr[ H\text{ is full rank}] &\notag \\
                               & \qquad \qquad +  \bE_{\L}[G(\L)| H\text{ is not full rank}]\Pr[ H\text{ is not full rank}] & \notag \\
                           &= \bE_{\L,X}\left[\frac{d^{2}(X,\L)}{n(\text{vol}(\L))^{2/n}}\Bigg| H\text{ is full rank}\right]\Pr[ H\text{ is full rank}] &\notag \\
                               & \qquad  +  \bE_{\L,X}\left[\frac{d^{2}(X,\L)}{n(\text{vol}(\L))^{2/n}}\Bigg| H\text{ is not full rank}\right]\Pr[ H\text{ is not full rank}] &\label{eq:msegoodness1} 
\end{align}
Since $p\Z^n\subset \L$, we have for every $\x\in\R^n$, $d(\x,\L)\leq d(\x,p\Z^n)\leq p\sqrt{n}/2$. Additionally, since $\L\subset \Z^n$, we have $\text{vol}(\L)\geq \text{vol}(\Z^n)=1$. 
Hence, we can say that for any Construction-A lattice,
\begin{equation}
 \frac{d^{2}(X,\L)}{n(\text{vol}(\L))^{2/n}}\leq \frac{p^{2}}{4}
 \label{eq:Glambda_ConstAbound}
\end{equation}
with probability $1$. Let $\delta$ be a positive constant that satisfies $\delta<\lba(B-2)-1$.
From the hypotheses of Theorem~\ref{theorem:LDA_MSEquantization}, we have
$B>2(1+R)/(1-R)$, and $\lba>1/R$. 
This guarantees that $\lba(B-2)-1>4/(1-R)-1>0$, and hence, we can choose a $\delta>0$.
 Using Lemma~\ref{lemma:LDA_fullrank}, we can bound $\Pr[ H\text{ is not full rank}]$ from above by $n^{-2\lba-\delta}$.
Using this and  (\ref{eq:Glambda_ConstAbound}) in (\ref{eq:msegoodness1}), and the fact that $\Pr[ H\text{ is full rank}]\leq 1$, we obtain
\begin{align}
 \bE_{\L}[G(\L)] &\leq \bE_{\L,X}\left[\frac{d^{2}(X,\L)}{n(\text{vol}(\L))^{2/n}}\Bigg| H\text{ is full rank}\right] + \frac{p^2}{4}\frac{1}{n^{2\lba+\dta}} &\notag \\
                   &= \bE_{\L,X}\left[\frac{d^{2}(X,\L)}{n(\text{vol}(\L))^{2/n}}\Bigg| H\text{ is full rank}\right] +o(1), &\label{eq:EGl_ubound}
\end{align}
thus completing the proof.\qed
\subsection{Proof of Lemma~\ref{lemma:dxL_bound}}
Recall that $r:=\reff(\L)$.
We want to show that for some $\delta>0$, the probability $\Pr[ d(\x,\L)>r( 1+n^{-\omega} )| H\text{ is full rank} ]$ goes to zero faster than $n^{-2\lba R+\delta}$.
The proof is along the same lines as di Pietro's proof of existence of lattices that achieve the capacity of the power constrained AWGN channel in~\cite{diPietrothesis}.
The parameters chosen in~\cite{diPietrothesis} were not sufficient to show that the lattices are good for MSE quantization.
We have adapted the proof to show that under stronger conditions (on the parameters of the lattice), we can obtain lattices which are good for MSE quantization.
For $\y\in\Z^n$, define
\[
 \xi_{\y}=\begin{cases}
              1 & \text{if } H\y\equiv \mathbf{0}\bmod p, \\
              0 & \text{otherwise.}
          \end{cases}
\]
Let $\rho=r(1+n^{-\omega})$. Recall that $\x+\rho\cB$ denotes an $n$-dimensional ball centered at $\x$ and having radius $\rho$. We define
\[
 X_\rho :=  \sum_{\y\in \Z^n\cap (\x+\rho\cB)} \xi_{\y},
\]
which is simply the number of  lattice points in $\x+\rho\cB$.
Let us define $\mathcal{E}(\rho)=|\Z^n\cap (\x+\rho\cB)|^2\frac{1}{p^{2n(1-R)}}$. From \cite[p.\ 119]{diPietrothesis}, we have
\begin{equation}
 \mathbb{E}[X_{\rho}] \geq \sqrt{\mathcal{E}(\rho)}.
\label{eq:EX_lbound}
\end{equation}

 In \cite[pp.\ 122--128]{diPietrothesis}, it was shown that the variance of $X_{\rho}$ can be bounded from above as follows.\footnote{The variance of $X_\rho$ is upper bounded by a sum of three terms, (\ref{eq:term1}), (\ref{eq:term2}), and (\ref{eq:term3}), 
 which was also studied in \cite{diPietrothesis} to show that nested LDA lattices achieve the capacity of the power constrained AWGN channel. We impose stronger constraints on $B$ and $\lba$ so as to ensure that (\ref{eq:dxL_bound}) goes to zero sufficiently fast as $n\to\infty$.}
\begin{align}
 \text{Var}(X_{\rho}) &\leq \sum_{s=1}^{\lfloor \frac{n(1-R)}{A+1-R}\rfloor}n^{s(2-\lba(A-2))} & \label{eq:term1} \\
                                     & \quad +\sum_{\substack{i,j,t \\j\leq n(1-R)/(B(1-R)+1)\\i+j+t=n(1-R)\\i+j>0}}\mathcal{E}(\rho)\left( 1+\frac{Bj}{n-Bj} \right)^{\frac{n-Bj+1}{2}}n^{j(1-\lba(B(1-R)-2))}  &\notag \\
                                     &\qquad\qquad\qquad  \times\left( 1+\frac{Bi}{n-Bi} \right)^{\frac{n-Bi+1}{2}} n^{i(1-\lba(B(1-R)-2))}(1+o(1)) &\label{eq:term2} \\
                                    &\quad +\sum_{\substack{i,j,t \\j\leq n(1-R)/(B(1-R)+1)\\i+j+t=n(1-R)\\i+j>0}}\mathcal{E}(\rho) \left( 1+\frac{Bj}{n-Bj} \right)^{\frac{n-Bj+1}{2}}n^{j\lba(2-B(1-R))} &\notag \\
                                                         &\hspace{4cm}\times n^{(j+t)\left(1+\lba\left(\frac{1}{AB-1}+\frac{1}{A}-1 \right)\right)}\frac{n^{\lba}}{\sqrt{\mathcal{E}(\rho)}}(1+o(1)). & \label{eq:term3}
\end{align}
We show that (\ref{eq:term1}), (\ref{eq:term2}) and (\ref{eq:term3}) are all bounded from above by $\mathcal{E}(\rho)n^{-2\lba R-\dta}(1+o(1))$.

Let 
\begin{equation}
  \delta:=\frac{1}{2}\min\{ \lba(A-2(1+R))-2, \: \lba(B(1-R)-2(1+R))-1 \}.
\label{eq:delta_defn} 
\end{equation}
 The hypotheses of Theorem~\ref{theorem:LDA_MSEquantization} ensure that $\dta>0$.

\subsubsection{The First Term, (\ref{eq:term1})}

We have
\[
 \sum_{s=1}^{\lfloor \frac{n(1-R)}{A+1-R}\rfloor}n^{s(2-\lba(A-2))}= n^{2-\lba(A-2)}(1+o(1)),
\]
provided that the exponent is negative.
As long as $2-\lba(A-2)<-2\lba R-\dta$, we have the first term bounded from above by $n^{-2\lba R-\dta}(1+o(1))$.
This condition is indeed satisfied, since by definition, $\dta<\lba(A-2(1+R))-2$.

\subsubsection{The Second Term, (\ref{eq:term2})}

For all $x>0$, we have $\ln(1+x)\leq x$, and hence
$(1+x)^{1/x}\leq e$. With this, we get
\[
  \left( 1+\frac{Bj}{n-Bj} \right)^{\frac{n-Bj}{2}} \leq e^{Bj/2}. 
\]
This implies that
\begin{align*} 
 \left( 1+\frac{Bj}{n-Bj} \right)^{\frac{n-Bj}{2}}n^{j(1-\lba(B(1-R)-2))} &\leq e^{Bj/2} n^{j(1-\lba(B(1-R)-2))} & \\
                                                                                                                          &= (c_1n)^{j(1-\lba(B(1-R)-2))}, &
\end{align*}
where $c_1= e^{B/(2(1-\lba(B(1-R)-2)))}$ is a positive constant. From (\ref{eq:delta_defn}), we have $\dta\leq \frac{1}{2}( \lba(B(1-R)-2(1+R))-1)$, and hence
$1-\lba(B(1-R)-2)\leq -2\lba R-2\dta$. Moreover, $c_1^{-2\lba R-2\dta}n^{-\dta}\leq 1$ for sufficiently large $n$. Hence,
\begin{equation}
  \left( 1+\frac{Bj}{n-Bj} \right)^{\frac{n-Bj}{2}}n^{j(1-\lba(B(1-R)-2))} \leq n^{j(-2\lba R-\dta)}
\end{equation}
for all sufficiently large $n$.
Similarly,
\begin{equation}
 \left( 1+\frac{Bi}{n-Bi} \right)^{\frac{n-Bi}{2}}n^{i(1-\lba(B(1-R)-2))} \leq n^{i(-2\lba R-\dta)}
\end{equation}
for all sufficiently large $n$.
Hence, the second term is bounded from above by
\begin{align}
 \sum_{\substack{i,j,t \\j\leq n(1-R)/(B(1-R)+1)\\i+j+t=n(1-R)\\i+j>0}}\mathcal{E}(\rho) n^{(i+j)(-2\lba R-\dta)} (1+o(1))&=\sum_{\substack{i,j \\j\leq n(1-R)/(B(1-R)+1)\\i+j\leq n(1-R)\\i+j>0}}\mathcal{E}(\rho) n^{(i+j)(-2\lba R-\dta)} (1+o(1))&\notag \\
                                                                                               &\leq \mathcal{E}(\rho) n^{-2\lba R-\dta} (1+o(1)). &\notag
\end{align}

\subsubsection{The Third Term, (\ref{eq:term3})}

Since $B>2/(1-R)$ and $\lba> 2\left(1-\frac{1}{AB-1}-\frac{1}{A}\right)^{-1}$, we have for $j\neq 0$,
\begin{equation}
  \left( 1+\frac{Bj}{n-Bj} \right)^{\frac{n-Bj+1}{2}}n^{j\lba(2-B(1-R))} = o(1), \text{ and}
\label{eq:thirdterm_1} 
\end{equation}
\begin{equation}
   n^{(j+t)\left(1+\lba\left(\frac{1}{AB-1}+\frac{1}{A}-1 \right)\right)} =o(1).
\label{eq:thirdterm_2} 
\end{equation}
If $j=0$, then the above terms are at most $1$.
Now,
\begin{align}
 \sqrt{\mathcal{E}(\rho)} &= |\Z^n\cap (\x+\rho\cB)|\frac{1}{p^{n(1-R)}} &\notag \\
                                           &\geq  V_n \left( \rho-\frac{\sqrt{n}}{2} \right)^n\frac{1}{p^{n(1-R)}} &\label{eq:thirdterm_mid1} \\
                                           & = V_n r^n \left( 1+\frac{1}{n^{\omega}} \right)^n \left( 1-\frac{\sqrt{n}}{2\rho} \right)^n  \frac{1}{p^{n(1-R)}}, &\notag 
\end{align}
where (\ref{eq:thirdterm_mid1}) follows from Lemma~\ref{lemma:Zn_cap_rB}.
But $V_nr^n=p^{n(1-R)}$. Using this, and simplifying, we get
\begin{align}
\sqrt{\mathcal{E}(\rho)}  & \geq p^{n(1-R)} \text{exp}\{n^{1-\omega}\} \text{exp}\left\{ \frac{\sqrt{2\pi e}}{2} n^{-\lba(1-R)}n (1+n^{-\omega})^{-1} \right\}\frac{1}{p^{n(1-R)}}(1+o(1)) &\notag \\
                                           & \geq \exp\{ n^{1-\omega}-o(1) \}. &\label{eq:thirdterm_3}
\end{align}
Therefore, $ 1/\sqrt{\mathcal{E}(\rho)}$ goes to zero faster than any polynomial. Combining (\ref{eq:thirdterm_1}), (\ref{eq:thirdterm_2}), and (\ref{eq:thirdterm_3}), 
we can conclude that (\ref{eq:term3}) is upper bounded by $\mathcal{E}(\rho)n^{-2\lba R -\dta}(1+o(1))$. 
As a consequence, the variance of $X_\rho$ is bounded from above by $3\mathcal{E}(\rho)n^{-2\lba R -\dta}(1+o(1))$.

\subsubsection{Proof of Lemma~\ref{lemma:dxL_bound}}

We have already seen in (\ref{eq:EX_lbound}) that $\bE[X_\rho]\geq \sqrt{\mathcal{E}(\rho)} $ and in the previous subsections, we showed that $\text{Var}(X_\rho)\leq \mathcal{E}(\rho)n^{-2\lba R -\dta}(1+o(1))$.
Therefore,
\begin{align}
 \Pr[d(\x,\L)>\rho] &= \Pr[X_\rho =0]\leq \Pr[X_\rho \leq 0]& \notag \\
                                & = \Pr\big[X_\rho-\bE [X_\rho]\leq -\bE [X_\rho]\big] & \notag \\
                                &\leq \Pr\big[|X_\rho - \bE [X_\rho]|\geq\bE[X_\rho]\big]. &\notag
\end{align}
Using Chebyshev's inequality, we get
\[
 \Pr[d(\x,\L)>\rho] \leq \frac{\text{Var}(X_\rho)}{\left( \bE[X_\rho] \right)^2} \leq \frac{3}{n^{2\lba R+\dta}}(1+o(1)),
\]
completing the proof of Lemma~\ref{lemma:dxL_bound}. \qed
\subsection{Proof of Lemma~\ref{lemma:ELX_EUL}}

 Recall that $U$ is uniformly distributed over $[0,p)^n$, and $X$ is uniformly distributed over $\cV(\L)$. We have,
\begin{align}
&\mathbb{E}_{U}\mathbb{E}_{\L}[d^2(U,\L)|H\text{ is full rank}] &\notag\\
				   &= \int_{\u\in[0,p)^n}\sum_{\L_1}d^2(\u,\L_1) \frac{Pr[\L=\L_1|H \text{ is full rank}]}{p^n} d\u &\notag \\
                                                                                                                     &= \sum_{\L_1}\int_{\u\in[0,p)^n}d^2(\u,\L_1) \frac{Pr[\L=\L_1|H \text{ is full rank}]}{p^n} d\u &\notag \\
                                                                                                                     &=\sum_{\L_1}\sum_{\z\in \L_1\cap [0,p)^n}\int_{\x\in\cV(\L_1)}d^2(\x+\z,\L_1) \frac{Pr[\L=\L_1|H \text{ is full rank}]}{p^n} d\x .&\notag 
\end{align}
 For all $\z\in\L$, we have $d(\x+\z,\L)=d(\x,\L)$. Hence,
\begin{align}
 \mathbb{E}_{U}\mathbb{E}_{\L}[d^2(U,\L)|H\text{ is full rank}]&=\sum_{\L_1}p^{nR}\int_{\x\in\cV(\L_1)}d^2(\x,\L_1) \frac{Pr[\L=\L_1|H \text{ is full rank}]}{p^n} d\x &\notag \\
                                                                                                                     &=\sum_{\L_1}\int_{\x\in\cV(\L_1)}d^2(\x,\L_1) \frac{Pr[\L=\L_1|H \text{ is full rank}]}{p^{n(1-R)}} d\x &\notag \\
                                                                                                                     &=\mathbb{E}_{\L}\mathbb{E}_{X}[d^2(X,\L)|H\text{ is full rank}]. &\notag
\end{align}
This completes the proof.\qed

\section*{Appendix~C}
The proof proceeds by splitting the summation in (\ref{eq:LDAdualpacking_step1}) into four parts, and showing that each quantity goes to zero as $n\to\infty$. The sum is divided into the following regimes:
\begin{enumerate}
 \item $1\leq t < \vartheta n(1-R)$,
 \item $ \vartheta n(1-R)\leq t < n(1-R)/2$,
 \item $n(1-R)/2 \leq t < (1-R-C_1/\ln n) n-1$,
 \item $(1-R-C_1/\ln n) n -1\leq t \leq n$,
\end{enumerate}
where $C_1$ is as defined in (\ref{eq:C1_defn}).
In each case, we will use the appropriate expansion properties of the underlying Tanner graph to prove the desired result.

\subsection{Case 1: {$1\leq t < \vartheta n(1-R)$}}
We will use property (R1) of the expander graph in this part of the proof. In this case, we have $t=|\mathrm{Supp}(\u)|\leq \vartheta n(1-R)$. Therefore, $|N(\mathrm{Supp}(\u))|=|\mathbb{S}(\u)|\geq B t$,
so that $\mathtt{1}_m(\Su)=0$ for $m<B t$. Consider
\begin{align}
 \phi_1(n) &:= \sum_{t=1}^{\vartheta n(1-R)}\sum_{\substack{\u\in\Fp^{n(1-R)}\\ |\mathrm{Supp}(\u)|=t}} \sum_{m=1}^{n}\mathtt{1}_m(\Su)
                                                     \frac{1}{p^m}\vert\{\x\in\Z^n\cap r_n\cB : |\mathrm{Supp}(\x)|\leq m\}\vert &\notag \\ 
                 &\leq \sum_{t=1}^{\vartheta n(1-R)}\sum_{\substack{\u\in\Fp^{n(1-R)}\\ |\mathrm{Supp}(\u)|=t}} \sum_{m=B t}^{n}
                                                     \frac{1}{p^m}\vert\{\x\in\Z^n\cap r_n\cB : |\mathrm{Supp}(\x)|\leq m\}\vert. &\notag 
\end{align}
Using Lemma~\ref{lemma:Zn_cap_rB}, the above quantity can be bounded from above as
\begin{align}
 \phi_1(n) &\leq \sum_{t=1}^{\vartheta n(1-R)}\sum_{\substack{\u\in\Fp^{n(1-R)}\\ |\mathrm{Supp}(\u)|=t}} \sum_{m=B t}^{n}
                                                     \frac{1}{p^m}\nchoosek{n}{m}V_m\left( r_n+\frac{\sqrt{m}}{2} \right)^m &\notag \\
                   &\leq \sum_{t=1}^{\vartheta n(1-R)}\sum_{\substack{\u\in\Fp^{n(1-R)}\\ |\mathrm{Supp}(\u)|=t}} \sum_{m=B t}^{n}
                                                     \frac{1}{p^m}\nchoosek{n}{m}V_m r_n^{m}\left( 1+\frac{\sqrt{m}}{2r_n} \right)^m &\notag \\
                    &= \sum_{t=1}^{\vartheta n(1-R)}\sum_{\substack{\u\in\Fp^{n(1-R)}\\ |\mathrm{Supp}(\u)|=t}} \sum_{m=B t}^{n}
                                                     \frac{1}{p^m}\nchoosek{n}{m}V_m \frac{p^{mR}}{V_n^{m/n}}\zeta_n^m\left( 1+\frac{\sqrt{m}}{2r_n} \right)^m .&\label{eq:phi1_1} 
\end{align}
Using Stirling's approximation, we get 
\[
  V_m=\frac{\pi^{m/2}}{\Gamma(1+m/2)} \leq \frac{\pi^{m/2}e^m}{(2\pi)^{1/2}m^{m+1/2}},
\]
and
\[
 V_n \geq \frac{\pi^{n/2}e^n}{en^{n+1/2}}.
\]
Therefore,
\begin{equation}
 \frac{V_m}{V_n^{m/n}}\leq c'\left(\frac{n}{m}\right)^{m+1/2}(1+o(1)),
 \label{eq:Vm_Vn_0}
\end{equation}
where $c'$ is a positive constant.
If $m>an$ for some $0<a<1$, then
\begin{equation}
 \frac{V_m}{V_n^{m/n}}\leq c\left(\frac{n}{m}\right)^{m}(1+o(1)),
 \label{eq:Vm_Vn}
\end{equation}
where $c$ is a positive constant.

Observe that $\zeta_n<1$ for all sufficiently large $n$, and $1+\frac{\sqrt{m}}{2r_n}\leq 2$. Using this, and (\ref{eq:Vm_Vn_0}) , the inequality (\ref{eq:phi1_1}) reduces to
\begin{align}
 \phi_1(n) &\leq c'\sum_{t=1}^{\vartheta n(1-R)}\sum_{\substack{\u\in\Fp^{n(1-R)}\\ |\mathrm{Supp}(\u)|=t}} \sum_{m=B t}^{n}
                                                     \frac{1}{p^{m(1-R)}}\nchoosek{n}{m}\left(2\frac{n}{m}\right)^m\left(\frac{n}{m}\right)^{1/2} (1+o(1)) &\notag \\
                   &\leq c'\sum_{t=1}^{\vartheta n(1-R)} \nchoosek{n(1-R)}{t}p^t  \sum_{m=B t}^{n}
                                                     \frac{1}{p^{m(1-R)}}\nchoosek{n}{m}\left(2\frac{n}{m}\right)^m\left(\frac{n}{m}\right)^{1/2} (1+o(1)). &\notag 
\end{align}
Using the inequalities $\nchoosek{n}{k}\leq n^k$ and $n/m\leq n$, we get
\begin{align}
 \phi_1(n)    &\leq c'\sum_{t=1}^{\vartheta n(1-R)} (n(1-R))^tp^t  \sum_{m=B t}^{n}
                                                     \frac{(2n^2)^{m}}{p^{m(1-R)}}n^{1/2} (1+o(1)) &\notag \\
                    &=  c'\sum_{t=1}^{\vartheta n(1-R)} (n(1-R))^tp^t  
                                                     \frac{(2n^2)^{Bt}}{p^{Bt(1-R)}}n^{1/2} (1+o(1)) &\notag \\
                    &= c'\sum_{t=1}^{\vartheta n(1-R)}(2^B(1-R))^t n^{t(1+\lba+2B-\lba B(1-R))}n^{1/2}(1+o(1))&\notag\\
                    &\leq c'\sum_{t=1}^{\vartheta n(1-R)}(2^B(1-R))^t n^{t(3/2+\lba+2B-\lba B(1-R))}(1+o(1)). &\label{eq:phi1_bound}
\end{align}
But we have $3/2+\lba+2B-\lba B(1-R)<0$, because the hypothesis of Theorem~\ref{theorem:LDA_dualpacking} guarantees that
$\lba>\frac{2B+3/2}{B(1-R)-1}$. Using the fact that $\sum_{t=a}^{b}n^t = n^{a}(1+o(1))$, we can conclude that (\ref{eq:phi1_bound}) is bounded from above by 
$(c''n)^{3/2+\lba+2B-\lba B(1-R)}(1+o(1))$ for some constant $c''$, and hence goes to zero as $n\to\infty$.

\subsection{Case 2: {$ \vartheta n(1-R)\leq t < n(1-R)/2$}}
 We will use property (R2) of the expander graph in this part of the proof. Since $|\mathrm{Supp}(\u)|= t < n(1-R)/2$, we have $|N(\mathrm{Supp}(\u))|=|\mathbb{S}(\u)|\geq \beta t$.
 Therefore, $\Pr[\mathbb{S}(\u)=m]=0$ for $m<\beta t$. Proceeding along the same lines as in the previous subsection, we get
 \begin{align}
  \phi_2(n) &:= \sum_{t=\vartheta n(1-R)}^{n(1-R)/2}\sum_{\substack{\u\in\Fp^{n(1-R)}\\ |\mathrm{Supp}(\u)|=t}} \sum_{m=1}^{n}\mathtt{1}_m(\Su)
                                                     \frac{1}{p^m}\vert\{\x\in\Z^n\cap r_n\cB : |\mathrm{Supp}(\x)|\leq m\}\vert &\notag \\ 
                   &\leq  \sum_{t=\vartheta n(1-R)}^{n(1-R)/2}\sum_{\substack{\u\in\Fp^{n(1-R)}\\ |\mathrm{Supp}(\u)|=t}} \sum_{m=\beta t}^{n}
                                                     \frac{1}{p^m}\nchoosek{n}{m}\frac{V_m}{V_n^{m/n}}p^{mR}\left(1+\frac{\sqrt{m}}{2r_n}\right)^m(1+o(1)). &\notag 
\end{align}
Using (\ref{eq:Vm_Vn}), and the inequalities $\nchoosek{n}{m}\leq 2^n$ and $1+\frac{\sqrt{m}}{2r_n}\leq 2$,   
\begin{align}
     \phi_2(n)&\leq  c\sum_{t=\vartheta n(1-R)}^{n(1-R)/2}\sum_{\substack{\u\in\Fp^{n(1-R)}\\ |\mathrm{Supp}(\u)|=t}} \sum_{m=\beta t}^{n}
                                                     \frac{1}{p^{m(1-R)}}\nchoosek{n}{m}\left(\frac{n}{m}  \right)^m\left(1+\frac{\sqrt{m}}{2r_n}\right)^m (1+o(1)) &\notag\\
                   &\leq c \sum_{t=\vartheta n(1-R)}^{n(1-R)/2}\sum_{\substack{\u\in\Fp^{n(1-R)}\\ |\mathrm{Supp}(\u)|=t}} \sum_{m=\beta t}^{n}
                                                     \frac{1}{p^{m(1-R)}}2^n\left(\frac{n}{m}  \right)^m2^m(1+o(1)). \notag
\end{align}
Since $n\geq m\geq \beta \vartheta n(1-R)$, we get
\begin{align}
       \phi_2(n) &\leq  c\sum_{t=\vartheta n(1-R)}^{n(1-R)/2}\sum_{\substack{\u\in\Fp^{n(1-R)}\\ |\mathrm{Supp}(\u)|=t}} \sum_{m=\beta t}^{n}
                                                     \frac{1}{p^{m(1-R)}}2^n\left(\frac{1}{\beta \vartheta (1-R)}  \right)^n 2^n(1+o(1)) &\notag\\
                   &\leq  c\sum_{t=\vartheta n(1-R)}^{n(1-R)/2}\nchoosek{n(1-R)}{t}p^t \sum_{m=\beta t}^{n}
                                                     \frac{1}{p^{m(1-R)}}\left(\frac{4}{\beta \vartheta (1-R)}  \right)^n (1+o(1)) &\notag\\
                   &\leq c \sum_{t=\vartheta n(1-R)}^{n(1-R)/2}2^{n(1-R)}p^t 
                                                     \frac{1}{p^{\beta t(1-R)}}\left(\frac{4}{\beta \vartheta (1-R)}  \right)^n (1+o(1)) &\notag\\
                   &\leq c 2^{n(1-R)} 
                                                     \frac{1}{p^{(\beta (1-R)-1)\vartheta n(1-R)}}\left(\frac{4}{\beta \vartheta (1-R)}  \right)^n (1+o(1)), &
 \end{align}
which goes to zero as $n\to\infty$, since $\beta>1/(1-R)$ from Definition~\ref{defn:goodgraphs}.

\subsection{Case 3: {$n(1-R)/2 \leq t < (1-R-C_1/\ln n) n-1$}}
 
We will use the following property of $(\alpha, A, \beta, B,\epsilon,\vartheta)$-good expander graphs: 
\begin{lemma}[\cite{diPietrothesis},Lemma~3.2]
If $S\subset V$ is such that $|N(S)|<n(1-R)/2$, then $|S|\leq |N(S)|/\alpha$.
\label{lemma:expander_left}
\end{lemma}
\begin{proof}
 Let us prove the contrapositive of the above statement. Suppose that $|S|>|N(S)|/\alpha$. Equivalently, $|N(S)|<\alpha |S|$. 
 This implies that $|S|>n(1-R)/(2\alpha)$, otherwise we would be in violation of property (L2) in Definition~\ref{defn:goodgraphs}.
 But from (L2), we have $|N(S)|\geq \alpha n(1-R)/(2\alpha)=n(1-R)/2$, and this completes the proof.
\end{proof}

\begin{figure}
\begin{center}
 \includegraphics[width=7cm]{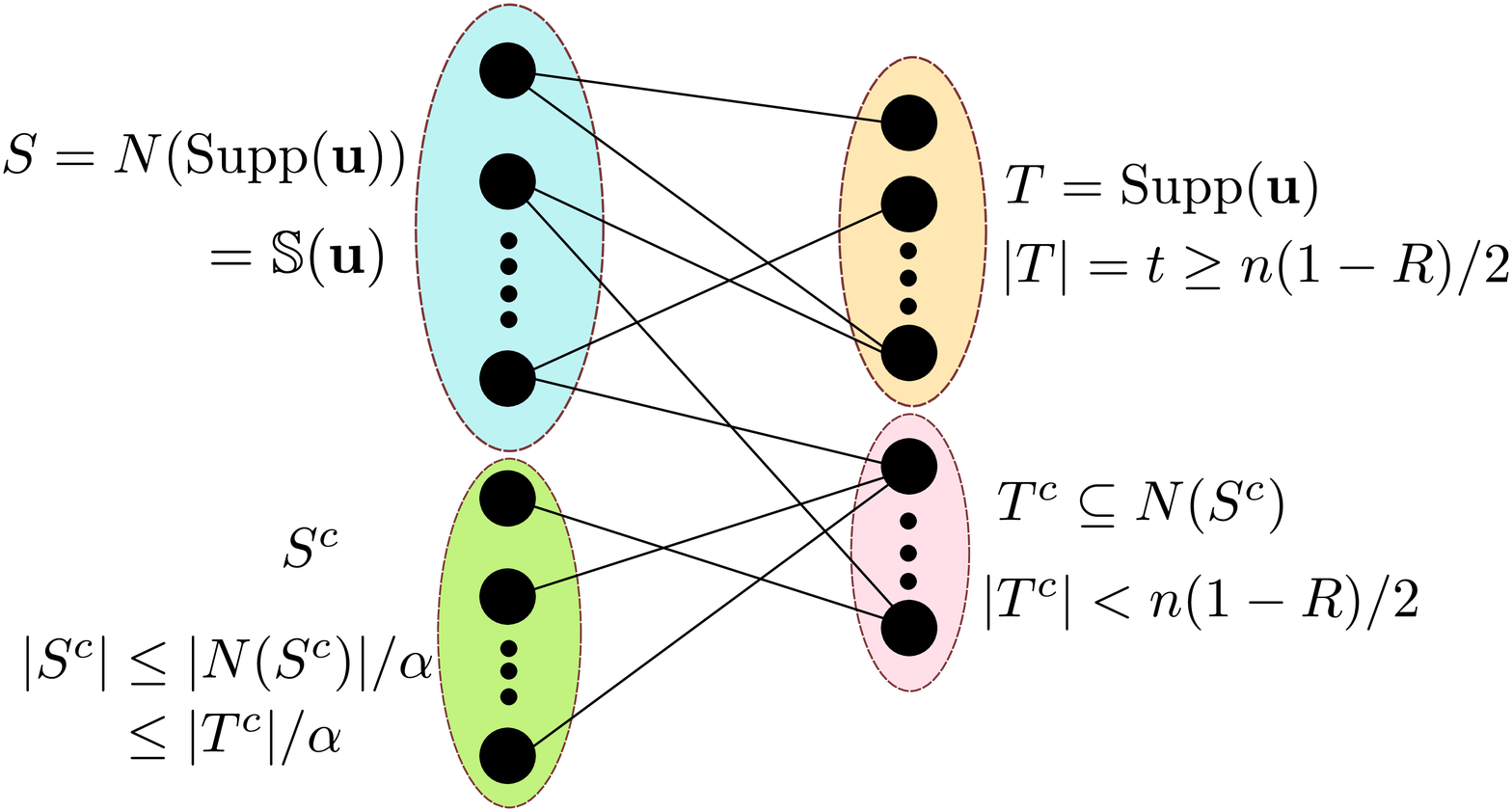}
\end{center}
\caption{Part 3 of proof.}
\label{fig:graph_part3}
\end{figure}
Since $T:=\mathrm{Supp}(\u)$ has at least  $n(1-R)/2$ vertices, the set $T^c$ has less than $n(1-R)/2$ vertices (see Fig.\ \ref{fig:graph_part3}). If $S:= \mathbb{S}(\u)=N(T)$, then, $S^c$ has does not have any neighbours from $T$.
Hence, $N(S^c)\subset T^c$. But $|T^c|<n(1-R)/2$ must imply that $|S^c|\leq |T^c|/\alpha$, from Lemma~\ref{lemma:expander_left}.
Therefore, $n-|S|\leq (n(1-R)-|T|)/\alpha$, or $|S|\geq n(1-(1-R)/\alpha)+t/\alpha$. This means that $\Pr[\mathrm{Supp}(\u)=m]=0$ for $m<n(1-(1-R)/\alpha)+t/\alpha$.

Consider
\begin{align}
     \phi_3(n) &:= \sum_{t= n(1-R)/2}^{n(1-R-C_1/\ln n)}\sum_{\substack{\u\in\Fp^{n(1-R)}\\ |\mathrm{Supp}(\u)|=t}} \sum_{m=1}^{n}\mathtt{1}_m(\Su)
                                                     \frac{1}{p^m}\vert\{\x\in\Z^n\cap r_n\cB : |\mathrm{Supp}(\x)|\leq m\}\vert &\notag \\ 
                      &\leq \sum_{t= n(1-R)/2}^{n(1-R-C_1/\ln n)}\sum_{\substack{\u\in\Fp^{n(1-R)}\\ |\mathrm{Supp}(\u)|=t}} \sum_{m=n(1-(1-R)/\alpha)+t/\alpha}^{n}
                                                     \frac{1}{p^m}\vert\{\x\in\Z^n\cap r_n\cB : |\mathrm{Supp}(\x)|\leq m\}\vert &\notag 
\end{align}
Following the approach in the previous subsections, the above reduces to
\begin{align}
   \phi_3(n) &\leq  c\sum_{t= n(1-R)/2}^{n(1-R-C_1/\ln n)}\nchoosek{n(1-R)}{t}p^t \sum_{m=n(1-(1-R)/\alpha)+t/\alpha}^{n}
                                                     \frac{1}{p^{m(1-R)}}\nchoosek{n}{m}\left(\frac{n}{m}\right)^m 2^m (1+o(1)) &\notag \\
                    & \leq  c\sum_{t= n(1-R)/2}^{n(1-R-C_1/\ln n)}8^n p^t \sum_{m=n(1-(1-R)/\alpha)+t/\alpha}^{n}
                                                     \frac{1}{p^{m(1-R)}}\left(\frac{n}{m}\right)^n  (1+o(1)), &\notag
\end{align}
where the last step uses the inequality $\nchoosek{n}{k}\leq 2^n$. Since $m\geq n(1-(1-R)/\alpha)+t/\alpha\geq n(1-(1-R)/\alpha +(1-R)/(2\alpha))$, we get 
\begin{align}
 \phi_3(n) &\leq c\sum_{t= n(1-R)/2}^{n(1-R-C_1/\ln n)}8^n p^t \sum_{m=n(1-(1-R)/\alpha)+t/\alpha}^{n}
                                                     \frac{1}{p^{m(1-R)}}\left(\frac{1}{1-(1-R)/\alpha +(1-R)/(2\alpha)}\right)^n  (1+o(1)) &\notag \\
                  &\leq c\sum_{t= n(1-R)/2}^{n(1-R-C_1/\ln n)} \left(\frac{8}{1-(1-R)/(2\alpha)}\right)^n \frac{p^t}{p^{n(1-R)(1-(1-R)/\alpha)+t/\alpha}}(1+o(1)) &\notag\\
                  &=  c\sum_{t= n(1-R)/2}^{n(1-R-C_1/\ln n)} n^{\frac{n\ln (8/(1-(1-R)/(2\alpha)))}{\ln( n)}} \frac{n^{\lba t}}{n^{\lba n(1-R)(1-(1-R)/\alpha)+\lba t/\alpha}}(1+o(1)) .&\label{eq:phi_3_2}
\end{align}
If we have
\[
     \lba n(1-R)\left( 1-\frac{1-R}{\alpha} \right) +\lba t\frac{(1-R)}{\alpha} -\lba t -\frac{n}{\ln n}\ln \left(\frac{8}{1-(1-R)/(2\alpha)}\right) > 1+\delta
\]
for some $\delta>0$, then (\ref{eq:phi_3_2}) is upper bounded by $c n\times n^{-1-\delta}(1+o(1))$, which goes to zero as $n\to\infty$.
Simplifying the above quantity gives us the condition
\[
  t< n(1-R)-n \frac{C_1}{\ln n} -\frac{1+\delta}{\lba(1-(1-R)/\alpha)},
\]
which is satisfied in this regime, and hence, $\phi_3(n)\to 0$ as $n\to\infty$.

\subsection{Case 4: {$(1-R-C_1/\ln n) n-1 \leq t < n$}}
For any subset of parity check nodes, $T\subset C$, we have $|N(T)|\geq |T|/(1-R)$. This is because the number of edges between $T$ and $N(T)$ is $|T|\Dta_V/(1-R)$, 
but the number of edges incident on each node in $N(T)$ from $T$  is at most $\Dta_V$. 
Therefore, we have
\begin{align}
 \phi_4(n) &:= \sum_{t= n(1-R-C_1/\ln n)}^{n}\sum_{\substack{\u\in\Fp^{n(1-R)}\\ |\mathrm{Supp}(\u)|=t}} \sum_{m=1}^{n}\mathtt{1}_m(\Su)
                                                     \frac{1}{p^m}\vert\{\x\in\Z^n\cap r_n\cB : |\mathrm{Supp}(\x)|\leq m\}\vert &\notag \\ 
                   &\leq c\sum_{t= n(1-R-C_1/\ln n)}^{n} \nchoosek{n(1-R)}{t}p^t \sum_{m=t/(1-R)}^{n}\frac{1}{p^{m(1-R)}} \nchoosek{n}{m}\left(\frac{n}{m}\right)^m\zeta_n^m (1+o(1))\notag \\
                   &= c\sum_{t= n(1-R-C_1/\ln n)}^{n} \nchoosek{n(1-R)}{n(1-R)-t}p^t \sum_{m=t/(1-R)}^{n}\frac{1}{p^{m(1-R)}} \nchoosek{n}{n-m}\left(\frac{n}{m}\right)^m\zeta_n^m (1+o(1)).\notag 
\end{align}
Since $\nchoosek{n}{n-k}$ is a decreasing function of $k$ for $k>n/2$, we have
\begin{align}
  \phi_4(n)&\leq c\sum_{t= n(1-R-C_1/\ln n)}^{n} \nchoosek{n(1-R)}{nC_1/\ln n}p^t \sum_{m=t/(1-R)}^{n}\frac{1}{p^{m(1-R)}} \nchoosek{n}{nC_1/((1-R)\ln n)}&\notag\\
  &\qquad\qquad \qquad\qquad \times\left(\frac{n}{n-\frac{nC_1}{(1-R)\ln n}}\right)^m\zeta_n^m (1+o(1)).\notag
\end{align}
Using the inequality $\nchoosek{n}{m}\leq \left(\frac{ne}{m}\right)^m$ and simplifying, we get
\begin{align}
 \phi_4(n) &\leq  c\sum_{t= n(1-R-C_1/\ln n)}^{n} \left(\frac{e(1-R)\ln n}{C_1}\right)^{nC_1/\ln n}p^t \sum_{m=t/(1-R)}^{n}\frac{1}{p^{m(1-R)}} &\notag \\
                 &\qquad\qquad \times\left(\frac{e(1-R)\ln n}{C_1}\right)^{nC_1/((1-R)\ln n)}\left(\frac{1}{1-\frac{C_1}{(1-R)\ln n}}\right)^n\zeta_n^m (1+o(1)).\notag
\end{align}
For all sufficiently large $n$, we have $m\geq n(1-C_1/((1-R)\ln n))>n/2$. Therefore, since $\zeta_n<1$, we have
\begin{align}
 \phi_4(n) &\leq  c\sum_{t= n(1-R-C_1/\ln n)}^{n} \left(\frac{e(1-R)\ln n}{C_1}\right)^{nC_1/\ln n}p^t \sum_{m=t/(1-R)}^{n}\frac{1}{p^{m(1-R)}} &\notag \\
                 &\qquad\qquad \times\left(\frac{e(1-R)\ln n}{C_1}\right)^{nC_1/((1-R)\ln n)}\left(\frac{1}{1-\frac{C_1}{(1-R)\ln n}}\right)^n\zeta_n^{n/2} (1+o(1)) &\notag \\
                 &\leq  c\sum_{t= n(1-R-C_1/\ln n)}^{n} \left(\frac{e(1-R)\ln n}{C_1}\right)^{2nC_1/((1-R)\ln n)}\left(\frac{1}{1-\frac{C_1}{(1-R)\ln n}}\right)^n\zeta_n^{n/2} (1+o(1))\notag\\
                 &\leq c n \left(\frac{e(1-R)\ln n}{C_1}\right)^{2nC_1/((1-R)\ln n)} \left(\frac{1}{1-\frac{C_1}{(1-R)\ln n}}\right)^n\zeta_n^{n/2} (1+o(1)),\notag
\end{align}
which goes to zero as $n\to\infty$ because of our choice of $\zeta_n$. This completes the proof of Theorem~\ref{theorem:LDA_dualpacking}. 
\qed

\end{document}